\documentclass[journal]{IEEEtran}
\usepackage{amsbsy}%,color,multicol}
\usepackage{floatflt} % text flows around figures

\usepackage{amsmath}
\usepackage{amssymb}
\usepackage{mathtools}
\usepackage{times}
\usepackage{graphics}
\usepackage{graphicx}
\usepackage{xspace}
\usepackage{paralist} % compact and in-paragraph* lists
\usepackage{setspace} % buggy somehow
\usepackage{xypic}
\xyoption{curve}
\usepackage{latexsym}
\usepackage{theorem}
\usepackage{ifthen}
\usepackage{subfigure}
\usepackage{booktabs}
\usepackage{algorithm}
\usepackage{algorithmic}
\usepackage{color}

\newcounter{brojac}

{\theoremheaderfont{\it} \theorembodyfont{\rmfamily}
\newtheorem{assumption}[brojac]{Assumption}
}

{\theoremheaderfont{\it} \theorembodyfont{\rmfamily}
\newtheorem{theorem}{Theorem}
\newtheorem{lemma}[theorem]{Lemma}

}

\title{}
\title{Convergence Rates of Distributed Nesterov-like Gradient Methods on Random Networks}

\author{Du$\check{\mbox{s}}$an Jakoveti\'c, Jo\~ao Xavier, and Jos\'e M.~F.~Moura$^{\star}$
\thanks{The first and second authors are with Instituto de Sistemas e Rob\'otica~(ISR), Instituto Superior T\'ecnico~(IST), Technical University of Lisbon, 1049-001 Lisbon, Portugal.  Their work is supported by: the Carnegie Mellon$|$Portugal Program under a grant from the Funda\c{c}\~ao de Ci$\hat{\mbox{e}}$ncia e Tecnologia~(FCT) from Portugal; by FCT grants CMU-PT/SIA/0026/2009 and SFRH/BD/33518/2008 (through the Carnegie Mellon$|$Portugal Program managed by ICTI); by ISR/IST plurianual funding (POSC program, FEDER). The first and third authors are with Department of Electrical and Computer Engineering, Carnegie Mellon University, Pittsburgh, PA 15213, USA. Their work is funded by AFOSR grant~FA95501010291 and by NSF grant~CCF1011903. Du$\breve{\mbox{s}}$an Jakoveti\'c holds a fellowship from~FCT. Authors e-mails:
[djakovetic,jxavier]@isr.ist.utl.pt, moura@ece.cmu.edu.}}
\begin{document}
\maketitle \thispagestyle{empty} \maketitle
\vspace{-1.2cm}

\begin{abstract}
We consider distributed optimization in random networks where $N$ nodes cooperatively
minimize the sum $\sum_{i=1}^N f_i(x)$ of their individual convex costs.
Existing literature
 proposes distributed gradient-like methods that are computationally
 cheap and resilient to link failures, but have slow convergence rates.
In this paper, we propose accelerated distributed gradient methods
that: 1) are resilient to link failures; 2) computationally cheap; and 3) improve convergence rates over other gradient methods.
%and establish their convergence rate guarantees on random networks,
%in terms of the expected optimality gap in the cost function at arbitrary node~$i$.
We model the network by a sequence of independent, identically distributed random matrices
$\{W(k)\}$ drawn from the set of symmetric, stochastic matrices with positive diagonals.
The network is connected on average and the cost functions are convex, differentiable, with Lipschitz continuous and bounded gradients.
We design two distributed Nesterov-like gradient methods that modify the D--NG and D--NC methods that we proposed for static networks. We prove their convergence rates in terms of the expected optimality gap at the cost function. Let $k$ and $\mathcal K$ be the number of per-node gradient evaluations
 and per-node communications, respectively. Then the modified D--NG achieves rates~$O(\log k/k)$ and $O(\log \mathcal K/\mathcal K)$, and the modified D--NC rates~$O(1/k^2)$ and $O(1/\mathcal K^{2-\xi})$, where $\xi>0$ is arbitrarily small. For comparison, the standard distributed gradient method cannot do better than $\Omega(1/k^{2/3})$ and $\Omega(1/\mathcal K^{2/3})$, on the same class of cost functions (even for static networks).
 % Finally, we extend the modified D--NC method to constrained optimization
  % with compact constraint sets.
   Simulation examples illustrate our analytical findings.
\end{abstract}
\hspace{.43cm}\textbf{Keywords:} Distributed optimization, convergence rate, random networks, Nesterov gradient, consensus.

%Consensus, weight optimization, correlated link failures, unconstrained optimization, sensor networks, switching topology, broadcast gossip.
%\newpage

%%%%%%%%%%%%%%%%%%%%%%%%%%%%%%%%%%%%%%%%%%%%%%%%%%%%%%%%%%%%%%%%%%%%%%%%%
%\title{  Research report   }
%\author{Du$\breve{\mbox{s}}$an Jakoveti\'c}
%%%%%%%%%%%%%%%%%%%%%%%%%%%%%%%%%%%%%%%%%%%%%%%%%%%%%%%%%%%%%%%%%%%%%%%%%
%\begin{document}
\maketitle \thispagestyle{empty} \maketitle
%
%
%We propose a distributed algorithm for solving constrained convex optimization problem across a wireless sensor network (WSN).
%
%
%\newpage
%
%
\section{Introduction}
\label{section-introduction}
We study distributed optimization where $N$ nodes
in a (sensor, multi-robot, or cognitive) network minimize the sum $\sum_{i=1}^N f_i(x)$
 of their individual costs subject
 to a global optimization variable~$x \in {\mathbb R}^d$.
 Each $f_i:\,{\mathbb R}^d \rightarrow \mathbb R$
  is convex and known only by node~$i$.
  The goal for each node is to estimate
  the vector $x^\star \in {\mathbb R}^d$ of common interest to all nodes.
  Each node $i$ acquires locally data $d_i$ that reveals partial knowledge on $x^\star$ and forms the cost function $f_i(x;\,d_i)$  of the global variable $x$.
  The nodes cooperate to find $x^\star$ that minimizes
  $\sum_{i=1}^N f_i(x;d_i)$.
  %Our objective is in distributed, iterative algorithms that solve the latter problem, where nodes over iterations $k$ communicate only with their immediate neighbors in the network.
 This setup has been studied in the context of many signal processing applications, including:
 1) distributed estimation in sensor networks, e.g.,~\cite{RibeiroADMM1,RibeiroADMM2};
 2) acoustic source localization, e.g., \cite{Rabbat};
 and 3) spectrum sensing for cognitive radio networks, e.g.,~\cite{bazerque_lasso,bazerque_sensing}.
% encompasses many networked systems applications,
% including distributed inference, e.g.,~\cite{SoummyaEst},
% and acoustic source localization, e.g.,~\cite{Rabbat}, in sensor networks.

For the above problem,
reference~\cite{arxivVersion}, see also~\cite{cdc-submitted,asilomar-nesterov}, presents
two distributed Nesterov-like gradient algorithms for static (non-random) networks, referred to
as D--NG~(Distributed Nesterov Gradient algorithm)
and D--NC~(Distributed Nesterov gradient with Consensus iterations).
The distributed gradient methods D--NG and D--NC significantly
improve the convergence rates over standard distributed gradient methods, e.g., \cite{nedic_T-AC,duchi}.

In this paper, we propose the mD--NG and mD--NC algorithms, which modify the D--NG and D--NC algorithms, and, beyond proving their convergence, we solve the much harder problem of establishing their convergence rate guarantees on~\emph{random networks}.
%We propose the modified D--NG and D--NC
%algorithms and establish the convergence rate guarantees for the two methods
% on random networks.
Randomness in networks may arise when inter-node links fail as with random packet dropouts in wireless
sensor networks, and when communication protocols are random like with the gossip protocol~\cite{BoydGossip}.
We model the network by a sequence of random independent, identically distributed~(i.i.d.) weight matrices~${W(k)}$ drawn from a set of symmetric, stochastic matrices with positive diagonals, and we assume that the network is connected on average (the graph supporting $\mathbb E \left[ W(k)\right]$ is connected). We establish the convergence
  rates of the expected optimality gap in the cost function (at any node~$i$)
%  \footnote{Note that we assume deterministic cost functions $f_i$'s; the randomness is only due to the underlying random networks.}
   of mD-NG and mD-NC, in terms of the number of per node gradient evaluations~$k$ and the number of per-node communications~$\mathcal K$, when the functions $f_i$ are convex and differentiable, with
   Lipschitz continuous and bounded gradients.
 We show that the modified methods achieve \emph{in expectation} the same rates that the methods in~\cite{arxivVersion} achieve on static networks, namely: mD--NG converges at rates~$O(\log k/k)$ and~$O(\log \mathcal K/\mathcal K)$, while mD--NC has rates~$O(1/k^2)$ and~$O(1/\mathcal K^{2-\xi})$, where $\xi$ is an arbitrarily small positive number. We explicitly give the convergence rate constants in terms of the number of nodes~$N$ and the  network statistics, more precisely, in terms of the quantity $\overline{\mu}:=\left(\| \mathbb E[W(k)^2]-J\|\right)^{1/2}$ (See ahead paragraph with heading Notation.)

We contrast D--NG and D--NC in~\cite{arxivVersion} with their modified variants, mD--NG and mD--NC, respectively. Simulations in  Section~\ref{section-simulation-example} show that D--NG may diverge when links fail,  while mD--NG converges, possibly at a slightly lower rate on static networks and requires an additional ($d$-dimensional) vector communication per iteration~$k$.
 Hence, mD--NG compromises slightly speed of convergence for robustness to link failures.

   Algorithm mD--NC has one inner consensus with $2d$-dimensional variables per outer iteration $k$, while D--NC has two consensus algorithms with $d$-dimensional variables. Both D--NC variants converge in our simulations when links fail, showing similar performance.

%Finally, we extend mD--NC to random networks and \emph{constrained minimization} $\min_{x \in \mathcal X} \sum_{i=1}^N f_i(x)$,
%where the constraint set~$\mathcal X$ is convex and compact, and the $f_i$'s are convex, differentiable, with Lipschitz continuous gradients, establishing for the random network model above the same convergence rates~$O(1/k^2)$ and~$O(1/\mathcal K^{2-\xi})$.

The analysis here differs from~\cite{arxivVersion}, since the dynamics of disagreements are different from the dynamics in~\cite{arxivVersion}. This requires novel bounds on certain products of time-varying matrices. By disagreement,
 we mean how different the solution estimates of distinct nodes are,  say $x_i(k)$ and $x_j(k)$ for nodes $i$ and $j$.
  %(see Lemmas~\ref{lemma-B-k-t-modified-D-NG}--\ref{lemma-upper-bound-on-norm-B-k-t}.)
 %Second, the constrained optimization problem we consider requires with respect to~\cite{arxivVersion} a novel investigation of an inexact \emph{projected} Nesterov gradient method.% (see Lemma~\ref{lemma-progress-one-iteration}~(a).)

\textbf{Brief comment on the literature}. There is increased interest
in distributed optimization and learning. Broadly, the literature considers two types of methods,
namely, batch processing,  e.g.,~\cite{nedic_T-AC,nedic-gossip,nedic_T-AC-private,Matei,Sonia-Martinez,bazerque_lasso,bazerque_sensing},
and online adaptive processing, e.g.,~\cite{SayedConf,SayedOptim,Theodoridis}.
With batch processing, data is acquired beforehand, and hence the $f_i$'s are known
before the algorithm runs. In contrast, with adaptive online processing,
nodes acquire new data at each iteration~$k$ of the distributed algorithm.
We consider here batch processing.

Distributed \emph{gradient methods} are, e.g., in \cite{nedic_T-AC,nedic-gossip,nedic_T-AC-private,nedic_novo,asu-new-jbg,asu-random,Karl-Johansson-Subgradient,Matei,duchi,Rabbat,Rabbat-Consensus-Dual-Avg,Sonia-Martinez}. References~\cite{nedic_T-AC,nedic-gossip,asu-new-jbg} proved convergence of their algorithms under deterministically time varying or random networks. Typically, $f_i$'s are convex,  non-differentiable, and with bounded gradients over the constraint set. Reference~\cite{duchi} establishes~$O\left(\log k / \sqrt{k}\right)$
convergence rate (with high probability) of a version of the distributed dual averaging method.
%It asssumes that $f_i$'s that are convex, possibly non-differentiable,
%
 We assume a more restricted class $\mathcal F$ of cost functions--$f_i$'s that are convex and have Lipschitz continuous and
 bounded gradients, but, in contradistinction, we establish strictly faster convergence rates--at least~$O(\log k/k)$ that are not achievable by standard distributed gradient methods~\cite{nedic_T-AC} on the same class~$\mathcal{F}$. Indeed, \cite{arxivVersion} shows that the method in~\cite{nedic_T-AC} cannot achieve a worst-case rate better than~$\Omega\left(1/k^{2/3}\right)$ on the same class~$\mathcal{F}$, even for static networks.
 Reference~\cite{AnnieChen} proposes an accelerated distributed proximal gradient method, which resembles our D--NC method for \emph{deterministically} time varying networks;
in contrast, we deal here with \emph{randomly varying networks}.
For a detailed comparison of D--NC with~\cite{AnnieChen}, we refer to~\cite{arxivVersion}.

 Distributed augmented Lagrangian or ordinary Lagrangian dual methods, are e.g., in~\cite{bazerque_lasso,cooperative-convex,Joao-Mota,Joao-Mota-2,Uday,Terelius,RibeiroADMM1,RibeiroADMM2}.
They have in general more complex iterations than gradient methods, but may have a lower total communication cost, e.g.,~\cite{Joao-Mota-2}. To our best knowledge, the convergence rates of such methods have not been established for random networks.

\textbf{Paper organization}.
 The next paragraph sets notation.
Section~\ref{section-mD-NG}
introduces the network and optimization models and presents mD--NG and its convergence rate, which is proved in Sections~\ref{section-intermediate-results} and~\ref{section-proofs-mD-NG}.
  Section~\ref{section-mD-NC} presents mD--NC and its convergence rate, proved in~Section~\ref{section-proofs-mD-NC}.
  Section~\ref{section-discussion} discusses extensions to our results. Section~\ref{section-simulation-example}
   illustrates  mD--NG and mD--NC  on a Huber loss example.
   We conclude in Section~\ref{section-conclusion-random}. Certain auxiliary arguments are in Appendices~\ref{app:A} and~\ref{subsection-appendix-moments}.
   %
   %
%
% section-algorithms-random}
%presents mD--NG and mD--NC, the modified D--NG and D--NC algorithms, and states the results on their convergence rates.
%Section~\ref{section-proofs-disagreements}
%proves the bounds on disagreements of different nodes' estimates.
% Section~\ref{section-conv-rate-proofs-random} proves
% the convergence rates of the proposed methods.
%%  Section~\ref{section-projected-D-NC} presents the projected D--NC method for constrained optimization and proves its convergence rate.
% Section~\ref{section-simulation-example} provides a simulation example with Huber losses. Finally, Section~\ref{section-conclusion-random} concludes the paper.

\textbf{Notation}. Denote by: ${\mathbb R}^d$ the $d$-dimensional real space:
$A_{lm}$ or $[A]_{lm}$  the $(l,m)$ entry of~$A$; $A^\top$ the transpose of~$A$;
$[a]_{l:m}$ the selection of the $l$-th, $(l+1)$-th, $\cdots$, $m$-th entries of vector $a$;
$I$, $0$, $\mathbf{1}$, and $e_i$, respectively, the identity matrix, the zero matrix, the column vector with unit entries, and the $i$-th column of $I$; $J$ the $N \times N$ ideal consensus matrix $J:=(1/N)\mathbf{1}\mathbf{1}^\top$; $\otimes $ the Kronecker product of matrices; $\| \cdot \|_l$ the vector (matrix) $l$-norm of its argument; $\|\cdot\|=\|\cdot\|_2$ the Euclidean (spectral) norm of its vector (matrix) argument ($\|\cdot\|$ also denotes the modulus of a scalar); $\lambda_i(\cdot)$ the $i$-th smallest \emph{in modulus} eigenvalue; $A \succ 0$ a positive definite Hermitian matrix $A$; $\lfloor a \rfloor$ the integer part of a real scalar $a$; $\nabla \phi(x)$ and $\nabla^2 \phi(x)$ the gradient and Hessian at $x$ of a twice differentiable function $\phi: {\mathbb R}^d \rightarrow {\mathbb R}$, $d \geq 1$;
$\mathbb P(\cdot)$ and $\mathbb E[\cdot]$
the probability and expectation, respectively; and
$\mathcal{I}_{\mathcal A}$ the indicator of event $\mathcal A$.
For two positive sequences $\eta_n$ and $\chi_n$, we have: $\eta_n = O(\chi_n)$ if $\limsup_{n \rightarrow \infty}\frac{\eta_n}{\chi_n}<\infty$; $\eta_n = \Omega (\chi_n)$ if $\liminf_{n \rightarrow \infty}\frac{\eta_n}{\chi_n}>0$; and $\eta_n=\Theta(\chi_n)$ if $\eta_n = O(\chi_n)$; and $\eta_n = \Omega(\chi_n)$.
\section{Algorithm $\mathrm{m}$D--NG}
\label{section-mD-NG}
Subsection~\ref{subsection-model} introduces the network and optimization models,
Subsection~\ref{subsection-m-D-NG} the mD--NG algorithm, and Subsection~\ref{subsection-conv-rate-statement-mD-NG}
 its convergence rate.

\subsection{Problem model}
\label{subsection-model}

\textbf{Random network model}.
%\label{subsection-network-model-random}
%
%
%
%
The network is random, due to link failures or
communication protocol used (e.g., gossip, \cite{BoydGossip,dimakiskarmourarabbatscaglione}.) It is defined by a sequence $\{W(k)\}_{k=1}^{\infty}$ of $N \times N$ random weight matrices.
 \begin{assumption}[Random network] We have:
 \label{assumption-random-weight-matrices}
 \begin{enumerate}[(a)]
 \item The sequence $\{W(k)\}_{k=1}^{\infty}$ is i.i.d.
 \item Almost surely (a.s.),  $W(k)$ are symmetric, stochastic, with strictly positive diagonal entries.
 \item There exists $\underline{w}>0$ such that, for all $i,j=1,\cdots ,N$,
 a.s. $W_{ij}(k) \notin (0,\underline{w})$.
 \end{enumerate}
 \end{assumption}
By Assumptions \ref{assumption-random-weight-matrices}~(b) and~(c), $W_{ii}(k) \geq \underline{w}$ a.s., $\forall i$; also,
  $W_{ij}(k)$, $i\neq j$, may be zero,  but if $W_{ij}(k)>0$ (nodes $i$ and $j$ communicate)
 it is non-negligible (at least $\underline{w}$).

Let $\overline{W}:=\mathbb E\left[ W(k)\right]$, the supergraph $\mathcal G = (\mathcal{N},E)$,
 $\mathcal N$ the set of $N$ nodes, and
$E=\left\{ \{i,j\}:\,i<j,\,\overline{W}_{ij}>0  \right\}$--$\mathcal G$ collects all realizable links, all pairs
$\{i,j\}$ for which $W_{ij}(k)>0$ with positive probability.

Assumption~\ref{assumption-random-weight-matrices} covers link failures. Here, each link $\{i,j\} \in E$ at time $k$ is Bernoulli: when it is one, $\{i,j\}$ is online (communication), and when it is zero, the link fails (no communication). The Bernoulli links are independent over time, but may be correlated in space. Possible weights are: 1) $i\neq j$, $\{i,j\} \in E$: $W_{ij}(k)=w_{ij}=1/N$, when $\{i,j\}$ is online, and $W_{ij}(k)=0$, else; 2) $i\neq j$, $\{i,j\} \notin E$: $W_{ij}(k) \equiv 0$; and 3) $W_{ii}(k)=1-\sum_{j \neq i}W_{ij}(k)$. As an alternative, when the link occurrence probabilities and their correlations are known, set the weights $w_{ij}$, $\{i,j\}\in E$, as the minimizers of $\overline{\mu}$ (See Section~\ref{section-simulation-example} and~\cite{weightOpt} for details.)

We further make the following Assumption.
\begin{assumption}[Network connectedness]
\label{assumption-connected-random}
$\mathcal G$ is connected.
\end{assumption}
Denote by $\widetilde{W}(k)=W(k)-J=-(1/N)\mathbf 1 \mathbf 1^\top$, by
\begin{equation}
\label{eqn-tilde-phi-modified-D-NG}
\widetilde{\Phi}(k,t)=\widetilde{W}(k)\cdots \widetilde{W}(t+2),\:\:t=0,1,\cdots ,k-2,
\end{equation}
and by $\widetilde{\Phi}(k,k-1)=I$.
 One can show that ${\overline{\mu}}:=\left( \| \mathbb E\left[ W^2(k) \right] - J\|\right)^{1/2}$
 is the square root of the second largest eigenvalue of $\mathbb E \left[ W^2(k)\right]$ and that, under Assumptions~\ref{assumption-random-weight-matrices} and~\ref{assumption-connected-random}, ${\overline{\mu}}<1$.
Lemma~\ref{lemma-decay-moments} (proof in Appendix~\ref{app:A}) shows that ${\overline{\mu}}$ characterizes the geometric decay of the first and second moments of $\widetilde{\Phi}(k,t)$.
\begin{lemma}
\label{lemma-decay-moments}
Let Assumptions~\ref{assumption-random-weight-matrices} and~\ref{assumption-connected-random} hold. Then:
\begin{align}
\label{eqn-theorem-phi}
\mathbb E \left[ \left\|\widetilde{\Phi}(k,t)\right\|\right] &\leq N\,{\overline{\mu}}^{k-t-1}\\
\label{eqn-theorem-phi-2}
\mathbb E \left[ \left\|\widetilde{\Phi}(k,t)^\top \widetilde{\Phi}(k,t) \right\|\right] &\leq N^2\,{\overline{\mu}}^{2(k-t-1)}\\
\label{eqn-theorem-phi-3}
\mathbb E \left[ \left\|\widetilde{\Phi}(k,s)^\top \widetilde{\Phi}(k,t) \right\|\right] &\leq N^3\,{\overline{\mu}}^{(k-t-1)+(k-s-1)},
\end{align}
for all $t,s=0,\cdots ,k-1,$ for all $k=1,2,\cdots $
\end{lemma}
The bounds in~\eqref{eqn-theorem-phi}-\eqref{eqn-theorem-phi-3} may be loose, but are enough to prove the results below and simplify the presentation.

For static networks, $W(k) \equiv W$, $W$
doubly stochastic, deterministic, symmetric, and ${\overline{\mu}}:=\|W-J\|$ equals the spectral gap, i.e.,
the modulus of the second largest (in modulus) eigenvalue of~$W$.
For static networks, the constants $N,N^2, N^3$ in~\eqref{eqn-theorem-phi}--\eqref{eqn-theorem-phi-3}
 are reduced to one.

\textbf{Optimization model}.
%\label{subsection-optimization-model-random}
 We now introduce the optimization model.
The nodes solve the unconstrained problem:
\begin{equation}
\label{eqn-opt-prob-original}
%\begin{array}[+]{ll}
\mbox{minimize} \:\: \sum_{i=1}^N f_i(x) =:f(x).
%\end{array}
\end{equation}
The function $f_i: {\mathbb R}^d \rightarrow {\mathbb R}$ is known only to node $i$. We impose the following three Assumptions.
\begin{assumption}[Solvability]
\label{assumption-f-i-s}
 There exists a solution $x^\star \in {\mathbb R}^d$ such that
 $f(x^\star)=f^\star:=\inf_{x \in {\mathbb R}^d} f(x)$.
 \end{assumption}
 \begin{assumption}[Lipschitz continuous gradient]
 \label{assumption-f-i-s-b}
 For all $i$, $f_i$ is convex and has Lipschitz continuous gradient with constant $L \in [0,\infty)$:
\[
\|\nabla f_i(x) - \nabla f_i(y)\| \leq L \|x-y\|,\:\:\:\forall x,y \in {\mathbb R}^d.
\]
\end{assumption}
%
%We denote by ${{x^\star}}$ a solution to~\eqref{eqn-opt-prob-original} and the optimal value $f^\star:=f({{x^\star}})$.
%
%
%In addition to Assumptions \ref{assumption-network} and \ref{assumption-f-i-s},
%we assume exclusively one of the following two assumptions (Assumption \ref{assumption-bdd-gradients} or Assumption \ref{assumption-linear-growth}.)
%
%
%
\begin{assumption}[Bounded gradients]
\label{assumption-bdd-gradients}
There exists a constant $ G \in [0,\infty)$ such that, $\forall i$, $\|\nabla f_i(x)\|\leq G$, $\forall x \in {\mathbb R}^d.$
\end{assumption}
%
%
%We further denote by $f^\star:=\$
%
%
%We comment on the optimization model in Assumptions~\ref{assumption-f-i-s}--\ref{assumption-bdd-gradients}.
Assumptions~\ref{assumption-f-i-s} and~\ref{assumption-f-i-s-b} are standard in gradient methods; in particular, Assumption~\ref{assumption-f-i-s-b} is precisely the Assumption required by the centralized Nesterov gradient method~\cite{Nesterov-Gradient}. Assumption~\ref{assumption-bdd-gradients} is not required in centralized Nesterov.
Reference~\cite{arxivVersion}
  demonstrates that
  (even on) static networks
  and a constant~$W(k) \equiv W$,
  the convergence rates of D--NG or of the standard distributed gradient method in~\cite{nedic_T-AC} become arbitrarily slow if Assumption~\ref{assumption-bdd-gradients} is violated.% (see~\cite{arxivVersion} for a precise statement.)
%   Section~\ref{section-projected-D-NC}
%   considers a different setup
%   with compact constraints.

    Examples of functions $f_i$ that obey Assumptions~\ref{assumption-f-i-s}--\ref{assumption-bdd-gradients} include the logistic loss for classification, \cite{BoydADMoM},
     the Huber loss in robust statistics, \cite{Blatt-Hero-Gauchman}, or
the ``fair'' function in robust statistics, $\phi: {\mathbb R} \mapsto {\mathbb R}$, $\phi(x) =
b_0^2 \left( {\frac{|x|}{b_0}} - \log \left( 1+\frac{|x|}{b_0}\right) \right)$, where $b_0>0$, \cite{Blatt-Hero-Gauchman}.
%We consider a different optimization model in Section~\ref{section-projected-D-NC}.
%odified D--NG algorithm, i.e., we let Assumptions~\ref{assumption-network} and~\ref{assumption-bdd-gradients} hold.

%\section{Fast distributed gradient algorithms}
%\label{section-algorithms-random}
%%
%This section presents our distributed gradient algorithms for random networks.
%Subsection~\ref{subsection-modified-D-NG}
%presents mD--NG, Subsection~\ref{subsection-D-NC-link-failures}
%presents mD--NC, and Subsection~\ref{subsection-rates-statements}
% states our results on the convergence rates of the two algorithms.
%
%
%
\subsection{Algorithm mD--NG for random networks}
\label{subsection-m-D-NG}
We modify  D--NG in~\cite{arxivVersion} to handle random networks.
Node $i$ maintains its solution estimate
$x_i(k)$ and auxiliary
variable $y_i(k)$, $k=0,1,\cdots $ It uses arbitrary initialization $x_i(0)=y_i(0) \in {\mathbb R}^d$
 and, for $k=1,2,\cdots $, performs the updates
\begin{align}
\label{eqn-D-NG-modified}
x_i(k) &= \sum_{j \in O_i(k)} W_{ij}(k)\,y_j(k-1)\\
&
\nonumber
\hspace{1cm}- \alpha_{k-1} \nabla f_i (y_i(k-1)) \\
\label{eqn-D-NG-modified-2}
y_i(k) &= (1+\beta_{k-1}) \,x_i(k) \\
&
\nonumber
\hspace{1cm}
- \beta_{k-1}\,\sum_{j \in O_i(k)} W_{ij}(k)\,x_j(k-1).
\end{align}
In~\eqref{eqn-D-NG-modified}--\eqref{eqn-D-NG-modified-2}, $O_i(k)=\{ j\in\{1,\cdots ,N\}:\,W_{ij}(k)>0\}$
is the (random) neighborhood of node $i$ (including node~$i$) at time~$k$.
 For $k=0,1,2,\cdots ,$ the step-size $\alpha_k$ is:
 \begin{equation}
 \label{eqn-alpha-k-link-failures}
 \alpha_k=c/(k+1),\:\:c \leq 1/(2L).
 \end{equation}
  All nodes know~$L$ (or its upper bound) beforehand to set $\alpha_k$ in~\eqref{eqn-alpha-k-link-failures}. Section~\ref{section-discussion} relaxes this
 requirement.
Let~$\beta_k$ be the sequence from centralized Nesterov, \cite{Nesterov-Gradient}:
 \begin{equation}
 \label{eqn-beta-k-link-failures}
 \beta_k=\frac{k}{k+3}.
 \end{equation}

%The mD--NG algorithm works as follows.
%At iteration $k$, node $i$ receives
%the variables $x_j(k-1)$ and $y_j(k-1)$
% from its current neighbors $j \in O_i(k)-\{i\}$,
%  and updates $x_i(k)$ and $y_i(k)$ via~\eqref{eqn-D-NG-modified} and~\eqref{eqn-D-NG-modified-2}.

% We borrow the choice~$\beta_k$ in~\eqref{eqn-alpha-k-link-failures}
%from the fast centralized Nesterov gradient method~\cite{Nesterov-Gradient}.
% The centralized Nesterov gradient method~\cite{Nesterov-Gradient}
% uses a constant step size $\alpha \leq 1/L$. We adopt the
% diminishing step-size $\alpha_k$

The mD--NG algorithm~\eqref{eqn-D-NG-modified}--\eqref{eqn-D-NG-modified-2}
differs from D--NG in~\cite{arxivVersion}
in step~\eqref{eqn-D-NG-modified-2}. With D--NG, nodes communicate only the variables $y_j(k-1)$'s;
with mD--NG, they also communicate the $x_j(k-1)$'s (see sum term in~\eqref{eqn-D-NG-modified-2}).
This modification allows for the robustness
to link failures. (See Theorems~\ref{theorem-consensus-D-NG-modified} and~\ref{theorem-convergence-rate-D-NG-modified}
and the simulations in Section~\ref{section-simulation-example}.)
Further, mD--NG does not require
the weight matrix to be positive definite,
as D--NG in~\cite{arxivVersion} does.

\textbf{Vector form}. Let $x(k):=(x_1(k)^\top,\cdots ,x_N(k)^\top)^\top$,
 $y(k):=(y_1(k)^\top,\cdots ,y_N(k)^\top)^\top$, and $F:{\mathbb R}^{N d} \rightarrow {\mathbb R}$,
 $F(x_1,\cdots ,x_N):=f_1(x_1)+\cdots +f_N(x_N)$. Then, for $k=1,2,\cdots $, with $x(0)=y(0) \in {\mathbb R}^{N d}$, ${W}(k) \otimes I$
 the Kronecker product of $W(k)$
  with the $d \times d$ identity~$I$, mD--NG in vector form is:
\begin{align}
\label{eqn-modified-D-NG-vector-form}
x(k) &= \left( {W}(k) \otimes I\right)\, y(k-1) \\
&
\nonumber
\hspace{1cm}- \alpha_{k-1} \,\nabla F(y(k-1))  \\
\label{eqn-modified-D-NG-vector-form-2}
y(k) &= (1+\beta_{k-1})\,x(k) \\
&
\nonumber
\hspace{1cm} -\beta_{k-1} \,\left({W}(k)\otimes I\right)\,x(k-1).
\end{align}

\textbf{Initialization}. For notation simplicity, without loss of generality (wlog), we assume, with all proposed methods,
that nodes initialize their estimates to the same values, i.e., $x_i(0)=y_i(0)=x_j(0)=y_j(0)$,  for all~$i,j$; for example, $x_i(0)=y_i(0)=x_j(0)=y_j(0)=0.$
%
%
%\textbf{The inexact oracle framework}.
%We will study the convergence of the modified D--NG algorithm
%  through the framework of the inexact centralized Nesterov gradient method in Section~\ref{section-inexact-oracle}.
%Define, as before, the global average
%$\overline{x}(k):=\frac{1}{N}\sum_{i=1}^N x_i(k)$,
% as well as the disagreements
% $\widetilde{x}_i(k):=x_i(k)-\overline{x}(k)$
%  and $\widetilde{x}(k):=(\widetilde{x}_1(k),\cdots ,\widetilde{x}_N(k))^\top$.
%  (We define analogously the quantities $\overline{y}(k)$, $\widetilde{y}_i(k)$, and $\widetilde{y}(k)$.)
%
%
%It is easy to show
%  that the global averages $\overline{x}(k)$, $\overline{y}(k)$
%   follow the \emph{same} equations in~\eqref{eqn-our-alg-CM}--\eqref{eqn-our-alg-CM-2}
%    as with the original D--NG method.
%Also, Lemmas~\ref{lemma-relate-inexact-oracle-with-our-alg}
%and~\ref{lemma-progress-one-iteration} continue to
%hold. The difference with respect to the original D--NG algorithm is in the
%disagreements $\widetilde{x}(k)$ and $\widetilde{y}(k)$; as we
%will see, they follow a different dynamics than the disagreements
%with the original D--NG. We note that
%the quantities $x(k),y(k),\widetilde{x}(k),\widetilde{y}(k)$
% are now random, due to the assumed random weight matrices $W(k)$.
% For the analysis of the expected optimality
% gap at each node
% $\mathbb E\left[ f(x_i(k)) - f^\star\right]$, we will
%  need to upper bound first and second moments
%  of the disagreements $\|\widetilde{x}(k)\|$ and $\|\widetilde{y}(k)\|$.
%
%
%
%
%
%
%
%
\subsection{Convergence rate of mD--NG}
\label{subsection-conv-rate-statement-mD-NG}
We state our convergence rate for mD--NG operating in random networks.
%
% Subsection~\ref{subsection-D-NG-modified-statements-of-results}
% considers the modified D--NG algorithm,
% while Subsection~\ref{subsection-D-NC-modified-statements-of-results}
% considers the modified D--NC algorithm.
 Proofs are in Section~\ref{section-proofs-mD-NG}.
We estimate the expected optimality gap in the cost
  at each node~$i$ normalized by~$N$, e.g., \cite{duchi,Rabbat-Consensus-Dual-Avg}: $\frac{1}{N}\mathbb E \left[f(x_i)-f^\star\right]$, where $x_i$ is node $i$'s solution at a certain stage of the algorithm. We study how node~$i$'s optimality gap decreases with: 1) the number~$k$ of iterations, or of per-node gradient evaluations;
and 2) the total number~$\mathcal{K}$ of $2\,d$-dimensional vector communications per node.
%With both methods, the number of outer iterations~$k$ equals the number of gradient evaluations per node.
With mD--NG, $k=\mathcal{K}$--at each~$k$, there is one and only one per-node $2\,d$-dimensional communication
and one per-node gradient evaluation. Not so with mD--NC, as we will see.
%We refer to~$\mathcal{K}$ as the number of communication rounds.
%
We establish for both methods the mean square convergence rate on the mean square disagreements of different node estimates in terms of~$k$ and~$\mathcal K$,
 showing that it converges to zero.

%\textbf{Algorithm D--NG}. We establish two types of results:~1)
% the mean square convergence decay rate of nodes' disagreements (how far apart
% are different nodes' estimates);
% and~2) the convergence rate for the expected normalized
% optimality gap~$\frac{1}{N}
%  \mathbb E \left[  f(x_i(k)) - f^\star\right]$
%   at all nodes~$i$.

Denote by: the network-wide global averages of the nodes' estimates be $\overline{x}(k):=\frac{1}{N}\sum_{i=1}^N x_i(k)$ and $\overline{y}(k):=\frac{1}{N}\sum_{i=1}^N y_i(k)$; the disagreements:~$\widetilde{x}_i(k)=x_i(k)-\overline{x}(k)$ and $\widetilde{x}(k)=\left(\widetilde{x}_1(k)^\top,\cdots ,\widetilde{x}_N(k)^\top\right)^\top$,
  and analogously for $\widetilde y_i(k)$ and~$\widetilde y(k)$; and $\widetilde{z}(k):=\left( \widetilde{y}(k)^\top,\,\widetilde{x}(k)^\top\right)^\top$.
We have the following Theorem on $\mathbb E \left[\|\widetilde{z}(k)\|\right]$
 and $\mathbb E \left[\|\widetilde{z}(k)\|^2\right]$.
 Note $\|\widetilde{x}(k)\| \leq \|\widetilde{z}(k)\|$, and so $\mathbb E \left[\|\widetilde{x}(k)\|\right] \leq \mathbb E \left[\|\widetilde{z}(k)\|\right]$
  and $\mathbb E \left[\|\widetilde{x}(k)\|^2\right] \leq \mathbb E \left[\|\widetilde{z}(k)\|^2\right]$.
 (Equivalent inequalities hold for $\widetilde{y}(k)$.)
  Recall also ${\overline{\mu}}$ in Lemma~\ref{lemma-decay-moments}. Theorem~\ref{theorem-consensus-D-NG-modified} states that
the mean square disagreement of different nodes' estimates converges to zero at rate~$1/k^2.$
\begin{theorem}
\label{theorem-consensus-D-NG-modified}
Consider mD--NG \eqref{eqn-D-NG-modified}--\eqref{eqn-beta-k-link-failures} under Assumptions
  \ref{assumption-random-weight-matrices}--\ref{assumption-bdd-gradients}.
Then, for all $k=1,2,\cdots $
\begin{align}
\label{eqn-theorem-norm-z-k}
{\mathbb E}\left[ \|\widetilde{z}(k)\|\right] &\leq \frac{50\,c\,N^{3/2}\,G}{(1-{\overline{\mu}})^2}\, \frac{1}{k}\\
\label{eqn-theorem-norm-z-k-squared}
{\mathbb E}\left[ \|\widetilde{z}(k)\|^2\right] &\leq  \frac{ 50^2\,c^2\,N^4 G^2} {(1-{\overline{\mu}})^4\, k^2}.
\end{align}
%
%
%where ${\overline{\mu}},{\overline{\mu}}$ are given in Lemma~\ref{lemma-decay-moments}.
\end{theorem}
%
%

%\textbf{D--NG: Convergence rate for the optimality gap}.
 Theorem~\ref{theorem-convergence-rate-D-NG-modified} establishes the convergence rate
 of mD--NG as~$O(\log k/k)$ (and $O(\log \mathcal K/\mathcal K)$).
\begin{theorem}
\label{theorem-convergence-rate-D-NG-modified}
Consider mD--NG \eqref{eqn-D-NG-modified}--\eqref{eqn-beta-k-link-failures} under Assumptions~\ref{assumption-random-weight-matrices}--\ref{assumption-bdd-gradients}.
 Let $\|\overline{x}(0)-{{x^\star}}\| \leq R$, $R \geq 0$. Then, at any node $i$,
 the expected normalized optimality gap
 $\frac{1}{N}\mathbb E \left[  f(x_i(k))-f^\star\right]$ is $O(\log k/k)$; more precisely:
\begin{align}
\label{eqn-theorem-D-NG-rate-1}
&
\hspace{1cm}
\frac{\mathbb E\left[ f(x_i(k))-f^\star \right]}{N}
 \leq \frac{2\, R^2}{c} \,\frac{1}{k}\\
 &
 \nonumber
 + \frac{50^2 \,c^2\, N^3\,L \,G^2}{(1-{\overline{\mu}})^4}\, \frac{1}{k} \sum_{t=1}^{k-1} \frac{(t+2)^2}{(t+1)t^2} +
\, \frac{50\,N^2\,c\, G^2}{(1-{\overline{\mu}})^2}\, \frac{1}{k} .
%  \\
%% \begin{eqnarray}
%% \label{eqn-opt-gap-for-scaling}
%% \frac{1}{N} \left( f(x_i(k))-f^\star\right)
% &\leq&\hspace{-2mm}
% \label{eqn-theorem-D-NG-rate-constant}
%\mathcal{C} \left( \frac{1}{k}\sum_{t=1}^{k} \frac{(t+2)^2}{(t+1)t^2}\right),\:\:\:\:
%\mathcal{C} =  \frac{2R^2}{c} + \frac{46,128 \,c^2 L \,G^2}{(1-{\overline{\mu}})^4} +  \frac{54\,c \,G^2}{(1-{\overline{\mu}})^2}.
%\label{eqn-constant-scaling}
 \end{align}
\end{theorem}
For static networks, $W(k) \equiv W$ and
${\overline{\mu}}:=\|W-J\|$, the factors $N$ and $N^3$
 above reduce to one. Theorem~\ref{theorem-convergence-rate-D-NG-modified} is
similar to Theorem~5~(a) in~\cite{arxivVersion}.
For static networks, the rate constant in Theorem~\ref{theorem-convergence-rate-D-NG-modified}
depends on ${\overline{\mu}}$ as $O\left(\frac{1}{(1-{\overline{\mu}})^4}\right)$; with Theorem~5~(a) in~\cite{arxivVersion},
it is~$O(\frac{1}{(1-{\overline{\mu}})^{3+\xi}})$, for~$\xi>$ arbitrarily
small. Hence, mD--NG  has worse theoretical constant,
 but is robust to link failures. (See Section~\ref{section-simulation-example}.)

%We now comment on Theorem~\ref{theorem-convergence-rate-D-NG-modified}.*******

%
\section{Intermediate results}
\label{section-intermediate-results}
We establishe intermediate results on certain scalar sums and the products of time-varying $2 \times 2$
 matrices that arise in the analysis of mD--NG and mD--NC.
%
%
%Subsection~\ref{subsection-intermediate}
%analyzes products of certain time-varying
%$2 \times 2$ matrices.
%These products play an important role in
% the disagreement dynamics of mD--NG and mD--NC.
% Subsections~\ref{subsection-proof-consensus-D-NG} and~\ref{subsection-proof-consensus-D-NC-random-net}
%  use the results in~\ref{subsection-intermediate} to prove
%  Theorems~\ref{theorem-consensus-D-NG-modified} and~\ref{theorem-consensus-D-NC-modified}, respectively.
%
%
%\subsection{Intermediate results}
%\label{subsection-intermediate}
%
%
%We consider certain (deterministic) scalar sums
%and products of~$2 \times 2$
%matrices $B(k)$ indexed by iteration counter~$k$.

\textbf{Scalar sums}. We have the following Lemma.
 \begin{lemma}
 \label{lemma-bounds-on-sums}
 Let $0<r <1$. Then, for all $k=1,2,\cdots $
 \begin{align}
 \label{eqn-prva-suma}
 \sum_{t=1}^{k} r^{t}\,t &\leq \frac{r}{(1-r)^2} \leq \frac{1}{(1-r)^2}\\
 \label{eqn-druga-suma}
 \sum_{t=0}^{k-1}r^{k-t-1}\frac{1}{t+1} &\leq \frac{1}{(1-r)^2\,k}.
 \end{align}
 \end{lemma}
\begin{proof} Let $\frac{d}{dr} h(r)$ be the derivative of $h(r)$. Then~\eqref{eqn-prva-suma} follows from:
{
\allowdisplaybreaks
\begin{align*}
\sum_{t=1}^k r^t\,t &= r\,\sum_{t=1}^k r^{t-1}t = r \, \frac{d}{d r} \left( \sum_{t=1}^k r^t\right) \\
&
= r \, \frac{d}{d r} \left(\, \frac{r-r^{k+1}}{1-r} \right)\\
&
= \frac{r\left(1-(k+1)r^{k}(1-r) -r^{k+1} \right)}{(1-r)^2} \\
&\leq \frac{r}{(1-r)^2},\:\:\:\:\forall k=1,2,\cdots
\end{align*}%
}
%Consider the term $r^t\,t$ for a fixed $t=1,\cdots ,k.$
% There holds:
% \begin{eqnarray*}
% r^t\,t = 2\,r^{t/2} \,r^{t/2} \frac{t}{2}
%        \leq 2\,r^{t/2}\,\left(\sup_{z \geq 0} r^{z} z\right)
%        \leq 2\,r^{t/2}\, \frac{1}{e(-\log r)},
% \end{eqnarray*}
%%
%where the last inequality follows by evaluating
%$\left(\sup_{z \geq 0} r^{z} z\right) = \frac{1}{e(-\log r)}$.
% Further, using the inequality $1/(- \log r) \leq 1/(1-r)$, $r \in [0,1)$,
% we obtain:
% \[
% r^t\,t \leq \frac{1}{1-r}\,r^{t/2},\:\:t=1,\cdots ,k.
% \]
%Applying the latter to~\eqref{eqn-prva-suma}, we get:
%%
%
%%
%%
%\begin{eqnarray*}
%\sum_{t=1}^k r^t\,t \leq \frac{1}{1-r}\,\sum_{t=1}^k (\sqrt{r})^t
%                    \leq \frac{1}{1-r} \frac{1}{1-\sqrt{r}} \leq \frac{2}{(1-r)^2},
%\end{eqnarray*}
%%
%%
%%
%where the last inequality uses $ \frac{1}{1-\sqrt{r}} \leq  \frac{2}{1-r}$.
% Thus, the bound in~\eqref{eqn-prva-suma}.

 To obtain~\eqref{eqn-druga-suma}, use~\eqref{eqn-prva-suma} and $k/(t+1) \leq k-t$, $\forall t=0,1,\cdots ,k-1$:
\begin{align*}
\sum_{t=0}^{k-1}r^{k-t-1}\frac{1}{t+1} &= \frac{1}{k}\,\sum_{t=0}^{k-1}r^{k-t-1}\frac{k}{t+1}\\
&\leq \frac{1}{k\,r}\,\sum_{t=0}^{k-1}r^{k-t}(k-t)
\leq \frac{1}{k}\,\frac{1}{(1-r)^2}.
\end{align*}
%
%Applying the bound in~\eqref{eqn-prva-suma} to the last equation,
%and using $\sum_{t=1}^{k-1} r^t \leq \frac{1}{1-r} \leq \frac{1}{(1-r)^2}$,
%the bound in~\eqref{eqn-druga-suma} follows.
\end{proof}

\textbf{Products of matrices}.
For $k=1,2,\cdots $, let~$B(k)$ be:
\begin{equation}
\label{eqn-B-k-definition}
B(k) := \left[ \begin{array}{cc}
(1+\beta_{k-1}) & -\beta_{k-1}  \\
1 & 0
 \end{array} \right],
 \end{equation}
with~$\beta_{k-1}$ in~\eqref{eqn-beta-k-link-failures}.
 The proofs of Theorems~\ref{theorem-consensus-D-NG-modified} and~\ref{theorem-consensus-D-NC-modified}
rely on products $\mathcal{B}(k,t)$. Let~$\mathcal{B}(k,-1):=I$ and:
\begin{equation}
\label{eqn-B-products}
\mathcal{B}(k,t)\!:= \!B(k)\!\cdots\! B(k-t), t=0,1,\!\cdots\!,k-2,
\end{equation}
%
%
%
%
%
%
%We first give a Lemma that
%explicitly calculates the product $\mathcal{B}(k,t)$, $t=1,2,\cdots ,k-2$, for $k \geq 3$.
%
%
%
%
\begin{lemma}[Products $\mathcal{B}(k,t)$]
\label{lemma-B-k-t-modified-D-NG}
Let $k \geq 3$, $\mathcal{B}(k,t)$ in~\eqref{eqn-B-products}, $a_t:=3/(t+3)$, $t=0,1,\cdots $, and:
{
\small{
\begin{equation}
\label{eqn-sigma-u-delta}
B_1\!\!:=\!\!\left[\! \begin{array}{cc}
2 & -1  \\
1 & \phantom{-}0
 \end{array}\! \right]\!,B_2\!\!:=\!\!\left[\! \begin{array}{cc}
1 & -1  \\
1 & -1
 \end{array}\! \right]\!,B_3\!\!:=\!\!\left[ \!\begin{array}{cc}
1 & -1  \\
0 & \phantom{-}0
 \end{array}\! \right].
\end{equation}
}
}
Then, for~$t=1,2,\cdots ,k-2$:
\[
\mathcal{B}(k,t) = B_1^{t+1} - \sigma_2(k,t) \,B_2 - \sigma_3(k,t)\,B_3,
\]
where
\begin{align}
\label{eqn-sigma-2-sigma-3}
&\sigma_2(k,t) = a_{k-t-1}\,t+a_{k-t}\beta_{k-t-1}\,(t-1)\\
&
\nonumber
\!+\!a_{k-t+1}\beta_{k-t}\beta_{k-t-1}\!(t-2)\!+\!\cdots\!+a_{k-2}\beta_{k-3}\!\cdots\!\beta_{k-t-1}\\
\label{eqn-sigma-2-sigma-3-2}
&\sigma_3(k,t) = a_{k-t-1} +a_{k-t}\beta_{k-t-1} \\
&
\nonumber
+a_{k-t+1}\beta_{k-t}\beta_{k-t-1} +\cdots +a_{k-2}\beta_{k-3}\cdots \beta_{k-t-1}.
\end{align}
\end{lemma}
%
%
%
%As it is typical with the results like Lemma~\ref{lemma-B-k-t-modified-D-NG}, it is not easy to ``guess'' the closed form expression of $\mathcal{B}(k,t)$. However, once we obtain the closed form expression, Lemma~\ref{lemma-B-k-t-modified-D-NG} can be easily proved using mathematical induction on $t=1,2,\cdots ,k-2$.
%
%
%
%
We establish bounds on the sums $\sigma_2(k,t)$ and $\sigma_3(k,t)$.
\begin{lemma}
\label{lemma-bounds-on-sigma-2-sigma-3}
Let $\sigma_2(k,t)$ and $\sigma_3(k,t)$ in~\eqref{eqn-sigma-2-sigma-3}-\eqref{eqn-sigma-2-sigma-3-2}, $t=1,\cdots ,k-2$, $k \geq 3.$
Then:
\begin{eqnarray}
\frac{t^2}{k+2} \leq \sigma_2(k,t) \leq t+1,\:\:\:\:\:\: 0 \leq \sigma_3(k,t) \leq 1.
\end{eqnarray}
\end{lemma}
\begin{proof} We prove each of the four inequalities above.

\textbf{Proof of the right inequality on $\sigma_2(k,t)$}.
By induction on $t=1,\cdots ,k-2$. The claim holds for $t=1$, since
$\sigma_2(k,1)=a_{k-2}=3/(k+1) \leq 1+1$, $\forall k$. Let it be true for some $t \geq 1$.
For $t=1,\cdots ,k-3$, write $\sigma_2(k,t)$ as:
 \begin{equation}
 \label{eqn-sigma2-recursively}
 \sigma_2(k,t+1) =a_{k-t-2}(t+1) + \beta_{k-t-2}\sigma_2(k,t).
 \end{equation}
 Using~\eqref{eqn-sigma2-recursively} and the induction hypothesis:
 $\sigma_2(k,t+1) \leq (t+1)a_{k-t-2}+\beta_{k-t-2}(t+1) = (a_{k-t-2}+\beta_{k-t-2})(t+1)=t+1 \leq t+2$.
 Thus, the right inequality on $\sigma_2(k,t)$.

\textbf{Proof of the left inequality on $\sigma_2(k,t)$}. Again, by
  induction on $t$. The claim holds for $t=1$, since:
  \[
  \sigma_2(k,1) = a_{k-2}=\frac{3}{k+1} \geq \frac{1^2}{k+2}.
  \]
   Let the claim be true for some $t \in \{1,2,\cdots ,k-3\}$, i.e.:
   \begin{equation}
   \label{eqn-induction-hypothesis-sigma-2}
   \sigma_2(k,t) \geq \frac{t^2}{k+2}.
   \end{equation}
   We show that $\sigma_2(k,t+1) \geq \frac{(t+1)^2}{k+2}.$
   Using~\eqref{eqn-sigma2-recursively}:
   \begin{align*}
   &\sigma_2(k,t+1) \geq a_{k-t-2}(t+1) + \beta_{k-t-2} \frac{t^2}{k+2}\\
   &
   = \frac{(t+1)^2}{k+2} + \frac{t(k-t)+(2 k+5 t+5)}{(k+2)(k-t+1)} \geq \frac{(t+1)^2}{k+2},
   \end{align*}
   where the last equality follows after algebraic manipulations.
   By induction, the last inequality completes the proof
   of the lower bound on $\sigma_2(k,t)$.

%   \textbf{Proof of the lower bound on $\sigma_3(k,t)$} is trivial.

   \textbf{Proof of bounds on $\sigma_3(k,t)$}. The lower bound is trivial. The upper bound follows by induction.
   For $t=1$:
   \[
   \sigma_3(k,1)=a_{k-2}+a_{k-1}\beta_{k-2} \leq a_{k-2}+\beta_{k-2}=1.
   \]
   Let the claim hold for some $t \in \{1,\cdots ,k-3\}$, i.e.:
   \[
   \sigma_3(k,t) \leq 1.
   \]
   From~\eqref{eqn-sigma-2-sigma-3}:
   \[
   \sigma_3(k,t+1) = \beta_{k-t-2} \,\sigma_3(k,t) + a_{k-t-2}.
   \]
   Thus, by the induction hypothesis:
   \[
   \sigma_3(k,t+1) \leq \beta_{k-t-2}+a_{k-t-2} \leq 1,
   \]
   completing the proof of the upper bound on $\sigma_3(k,t)$.
\end{proof}
We upper bound $\|\mathcal{B}(k,k-t-2)\|$ of direct use in Theorem~\ref{theorem-consensus-D-NG-modified}.
\begin{lemma}
\label{lemma-upper-bound-on-norm-B-k-t}
Consider $\mathcal{B}(k,t)$ in~\eqref{eqn-B-products}. Then,
for all $t=0,\cdots ,k-1$, for all $k=1,2,\cdots $
\begin{equation}
\label{eqn-lemma-inequality-norm-b-k-t}
\| \mathcal{B}(k,k-t-2)\| \leq 8 \frac{(k-t-1)(t+1)}{k} + 5.
\end{equation}
\end{lemma}
\begin{proof}
Fix some $t \in \{1,\cdots ,k-2\}$, $k \geq 3$, and consider $\mathcal{B}(k,t)$ in Lemma~\ref{lemma-B-k-t-modified-D-NG}. It follows
 $B_1^t = t\,B_2+I$. Thus,
\begin{equation}
\label{eqn-B-K-t-proof}
\mathcal{B}(k,t) = \left(  t+1 - \sigma_2(k,t) \right)\,B_2 + I -\sigma_3(k,t) B_3.
\end{equation}
By Lemma~\ref{lemma-bounds-on-sigma-2-sigma-3},
the term:
\[
0 \leq   t+1 - \sigma_2(k,t) \leq t+1-t^2/(k+2).
\]
Using in~\eqref{eqn-B-K-t-proof} this equation, $\sigma_3(k,t) \leq 1$ (by Lemma~\ref{lemma-bounds-on-sigma-2-sigma-3}),
$\|B_2\| =2$, and $\|B_3\|=\sqrt{2} < 2$, get:
\begin{align}
\label{eqn-b-k-proof-2}
\|\mathcal{B}(k,t)\| \!\!\leq \!\!2\!\!\left(\!\!t\!\!+\!\!1\!\!-\!\!\frac{t^2}{k+2}\!\!\right)\!\!+\!\!3 &\!\!=\!\! 2\left(\!\!t-\frac{t^2}{k+2}\!\!\right)\!\! +\!\!5,
\end{align}
for all $t=1,2,\cdots ,k-2$, $k \geq 3$.
Next, from~\eqref{eqn-b-k-proof-2}, for $t=0,\cdots ,k-3$, $k \geq 3$, get:
\begin{align*}
&\|\mathcal{B}(k,k-t-2)\| \leq 2 \left( k-t-2-\frac{(k-t-2)^2}{k+2}  \right) + 5\\
&
= 2(k-t-2)\frac{t+4}{k+2}+5
\leq 8 (k-t-1)\frac{t+1}{k} +5,
\end{align*}
We used $(t+4)/(k+2) \leq 4(t+1)/k$ and proved~\eqref{eqn-lemma-inequality-norm-b-k-t}
for $t=0,\cdots ,k-3$, for $k \geq 3$. To complete the proof,
we show that~\eqref{eqn-lemma-inequality-norm-b-k-t}
 holds also for: 1) $t=k-2$, $k \geq 2$;
 2) $t=k-1$, $k \geq 1$.
Consider first case~1 and $\mathcal{B}(k,k-2)=B(k-1)=B_1-a_{k-1}B_3$, $k \geq 2$.
 We have $\|\mathcal{B}(k,k-2)\|\leq \|B_1\|+\|B_3\|<5$, and so~\eqref{eqn-lemma-inequality-norm-b-k-t}
 holds for $t=k-2$, $k \geq 2.$ Next, consider case~2
 and $\mathcal{B}(k,k-1)=I$, $k \geq 1$. We have that $\|\mathcal{B}(k,k-1)\|=1<5$,
 and so~\eqref{eqn-lemma-inequality-norm-b-k-t}
  also holds for $t=k-1$, $k \geq 1.$ This proves the Lemma.
\end{proof}
%
%
%
%
%
%
%
%

 %
 %
 %
 %
 %
% %
% %
%\section{Proofs of the main results}
%\label{section-proofs-modified-D-NG-D-NC}
%%
%%
%Current section uses the results in Sections~\ref{section-intermediate} and~\ref{section-inexact-oracle}
%to prove our main results, Theorems~\ref{theorem-consensus-D-NG-modified}--\ref{theorem-conv-rate-D-NC-modified}.
%Subsection~\ref{subsection-proofs-D-NG-modif} studies the modified D--NG algorithm,
%while Subsection~\ref{subsection-proofs-D-NC-random-net}
% studies the modified D--NC algorithm.
%
%
%

\section{Proofs of Theorems~\ref{theorem-consensus-D-NG-modified} and~\ref{theorem-convergence-rate-D-NG-modified}}
\label{section-proofs-mD-NG}
%\label{subsection-disagreement-modified-D-NG}
%
%
Subsection~\ref{subsection-proof-disagreement-mD-NG} proves Theorem~\ref{theorem-consensus-D-NG-modified},
while Subsection~\ref{subsection-proof-conv-rate-moD-NG} proves Theorem~\ref{theorem-convergence-rate-D-NG-modified}.

\subsection{Proof of Theorem~\ref{theorem-consensus-D-NG-modified}}
\label{subsection-proof-disagreement-mD-NG}
%
%
%We proceed by proving Theorem~\ref{theorem-consensus-D-NG-modified}.
Through this proof and the rest of the paper,
we establish certain equalities and inequalities on random
quantities of interest.
These equalities and inequalities further ahead
hold either: 1) surely, for any random
realization, or: 2) in expectation. From the notation, it is clear which
of the two cases is in force. For notational simplicity, we perform the proof of
Theorem~\ref{theorem-consensus-D-NG-modified} for the case $d=1$,
 but the proof extends for generic $d>1.$
The proof has three steps. In Step~1,
we derive the dynamic equation for the disagreement
$\widetilde{z}(k) = \left( \widetilde{y}(k)^\top,\,\widetilde{x}(k)^\top  \right)^\top$.
 In Step~2, we
 unwind the dynamic equation, expressing
 $\widetilde{z}(k)$
  in terms of the products $\widetilde{\Phi}(k,t)$ in~\eqref{eqn-tilde-phi-modified-D-NG}
   and $\mathcal{B}(k,t)$ in~\eqref{eqn-B-products}.
  Finally, in Step~3,
  we apply the already
  established bounds on the norms of the latter products.

\textbf{Step~1. Disagreement dynamics}.
Let $\widetilde{z}(k):=\left(\widetilde{y}(k)^\top,\widetilde{x}(k)^\top\right)^\top$.
Multiplying~\eqref{eqn-D-NG-modified}--\eqref{eqn-D-NG-modified-2} from the left by $(I-J)$, using
$(I-J)W(k)=\widetilde{W}(k)-J$, obtain:
\begin{eqnarray}
\label{eqn-recursion-D-NG-modified}
\widetilde{z}(k) = \left( B(k) \otimes \widetilde{W}(k) \right)\,\widetilde{z}(k-1) +u(k-1),
\end{eqnarray}
for $k=1,2,\cdots$ and $\widetilde{z}(0)=0$,
where
{\small{
\begin{align}
\label{eqn-u-k-modified}
u(k-1) \!\!=\!\! -\!\!\left[\!\!\!\! \begin{array}{cc} \alpha_{k-1}(1+\beta_{k-1})(I-J)\nabla F(y(k-1))\\
\alpha_{k-1} (I-J)\nabla F(y(k-1)) \end{array}\!\!\!\! \right]\!\!.
\end{align}
}
}

\textbf{Step~2. Unwinding recursion~\eqref{eqn-recursion-D-NG-modified}}.
Recall $\widetilde{\Phi}(k,t)$ in~\eqref{eqn-tilde-phi-modified-D-NG}, and $\mathcal{B}(k,t)$ in~\eqref{eqn-B-products}.
Then, unwinding~\eqref{eqn-recursion-D-NG-modified},
and using the Kronecker product property $(A \otimes B)(C \otimes D)
=(A B) \otimes (C D)$, we obtain for all $k=1,2,\cdots$
\begin{equation}
\label{eqn-unwinded-D-NG-modified}
\widetilde{z}(k) = \sum_{t=0}^{k-1} \left( \mathcal{B}(k,k-t-2) \otimes \widetilde{\Phi}(k,t) \right)\,u(t),
\end{equation}
The quantities
$u(t)$ and $\widetilde{\Phi}(k,t)$ in~\eqref{eqn-unwinded-D-NG-modified}
are random, while the $\mathcal{B}(k,k-t-2)$'s are deterministic.

\textbf{Step~3. Finalizing the proof}.
%Recall the unwinded expression for $\widetilde{z}(k)$ in~\eqref{eqn-unwinded-D-NG-modified}.
 Consider $u(t)$ in~\eqref{eqn-u-k-modified}.
By Assumption~\ref{assumption-bdd-gradients}, we have
$\|\nabla F(y(t))\| \leq \sqrt{N} G$.
Using this, the step-size $\alpha_t=c/(t+1)$,
and $\|I-J\|=1$, get $\|u(t)\| \leq \frac{\sqrt{3}\,c\,\sqrt{N}\,G}{t+1}$, for any
random realization of~$u(t)$.
With this bound, Lemma~\ref{lemma-upper-bound-on-norm-B-k-t},
and the sub-multiplicative and sub-additive properties of norms, obtain from~\eqref{eqn-unwinded-D-NG-modified}:
\begin{align*}
\|\widetilde{z}(k)\| &\leq \left( 8 \sqrt{3} \,c\, \sqrt{N}\,G \right)\,
\frac{1}{k}\,\sum_{t=0}^{k-1} \|\widetilde{\Phi}(k,t)\| \, (k-t-1) \\
&
+
\left( 5 \sqrt{3} \,c\, \sqrt{N}\,G \right)\,
\sum_{t=0}^{k-1} \|\widetilde{\Phi}(k,t)\| \, \frac{1}{t+1}. %\nonumber
\end{align*}
Taking expectation, and using Lemma~\ref{lemma-decay-moments}:
\begin{align*}
\mathbb E \left[ \|\widetilde{z}(k)\| \right] &\leq \left( 8 \sqrt{3} \,c\, N^{3/2}G \right)
\frac{1}{k}\sum_{t=0}^{k-1} {\overline{\mu}}^{k-t-1} (k-t-1)\\
&
+
\left( 5 \sqrt{3} \,c\, N^{3/2}\,G \right)\,
\sum_{t=0}^{k-1} {\overline{\mu}}^{k-t-1} \, \frac{1}{t+1}. %\nonumber
\end{align*}
Finally, applying Lemma~\ref{lemma-bounds-on-sums} to the last equation with $r=\overline{\mu}$,
the result in~\eqref{eqn-theorem-norm-z-k} follows.

Now prove~\eqref{eqn-theorem-norm-z-k-squared}.
Consider $\|\widetilde{z}(k)\|^2.$
From~\eqref{eqn-unwinded-D-NG-modified}:
\begin{align*}
&\|\widetilde{z}(k)\|^2
 \!\!= \!\!
 \sum_{t=0}^{k-1}\sum_{s=0}^{k-1}
 u(t)^\top  \!\!\left( \mathcal{B}(k,k-t-2)^\top \otimes \widetilde{\Phi}(k,t)^\top \right)\\
 &
 \left(
 \mathcal{B}(k,k-t-2) \otimes \widetilde{\Phi}(k,s)\right)\,u(s) \\
 &=
 \sum_{t=0}^{k-1}\sum_{s=0}^{k-1}
 u(t)^\top \,\left( \mathcal{B}(k,k-t-2)^\top \, \mathcal{B}(k,k-s-2) \right)\\
 &
 \otimes\,\left(
 \widetilde{\Phi}(k,t)^\top \, \widetilde{\Phi}(k,s)\right)\,u(s),
\end{align*}
where the last equality again uses the property $(A \otimes B)(C \otimes D)
= (A C) \otimes (B D)$. By the sub-additive and sub-multiplicative properties of norms, obtain:
{\small
{
\begin{align}
\label{eqn-recall-for-appendix}
&\|\widetilde{z}(k)\|^2
 \!\!\leq \!\!
\sum_{t=0}^{k-1}\sum_{s=0}^{k-1}\!\!
 \left\| \mathcal{B}(k,k-t-2)\right\| \!\! \left\|\mathcal{B}(k,k-s-2) \right\|\!\!\\
 &
 \nonumber
\hspace{1cm}
 \left\| \widetilde{\Phi}(k,t)^\top \, \widetilde{\Phi}(k,s)\right\|
 \,\|u(t)\|\,\|u(s)\| \\
 &
 \nonumber
 \leq\!\!
 \sum_{t=0}^{k-1}\sum_{s=0}^{k-1}\!\!
 \left(\!\! \frac{8(k-t-1)(t+1)}{k} \!\!+\!\! 5\!\!\right)\!\!\!\left(\!\! \frac{8(k-s-1)(s+1)}{k} + 5 \!\!\right)\\
 &
 \hspace{1cm}
 \left\| \widetilde{\Phi}(k,t)^\top \, \widetilde{\Phi}(k,s)    \right\|\,\frac{3\,c^2\,N G^2}{(t+1)(s+1)}. \nonumber
\end{align}
}
}
The last inequality uses Lemma~\ref{lemma-upper-bound-on-norm-B-k-t}
 and $\|u(\sqrt{}t)\| \leq \left(\sqrt{3} c\sqrt{N} G\right)/(t+1).$
 Taking expectation and applying Lemma~\ref{lemma-decay-moments}, obtain:
 {
 \allowdisplaybreaks
\begin{align*}
 &\mathbb E\! \left[ \|\widetilde{z}(k)\|^2   \right]\! \leq \!  \left(3c^2N^4 G^2\right)
 \!\!\sum_{t=0}^{k-1}\!\sum_{s=0}^{k-1}\!\!
 \left(\!\! \frac{8(k-t-1)(t+1)}{k}\!\! + \!\!5\!\!\right)\\
 &
 \left( \frac{8(k-s-1)(s+1)}{k} + 5 \right)
 \frac{\overline{\mu}^{k-t-1+k-s-1}}{(t+1)(s+1)}\\
 &=
 \left(\!3c^2N^4 G^2\!\right)\!\!\!
 \left(\!\sum_{t=0}^{k-1} \!\! \left(\! \frac{8(k-t-1)(t+1)}{k}\! + \!5\!\right) \!\!\frac{ {\overline{\mu}}^{k-t-1} }{t+1}   \!  \right)^2 \\
 &\leq
 \frac{ 50^2\,c^2\,N^4 G^2} {(1-{\overline{\mu}})^4\, k^2}.
 \end{align*}
 }
The last inequality applies Lemma~\ref{lemma-bounds-on-sums}.
 Thus, the bound in~\eqref{eqn-theorem-norm-z-k-squared}. The proof of Theorem~\ref{theorem-consensus-D-NG-modified} is complete.
\subsection{Proof of Theorem~\ref{theorem-convergence-rate-D-NG-modified}}
\label{subsection-proof-conv-rate-moD-NG}
%
%
%
%\textbf{Proof outline for Theorem~\ref{theorem-convergence-rate-D-NG-modified}}.
The proof parallels that of Theorem~5~(a) in~\cite{arxivVersion}. We outline it and refer to~(\cite{arxivVersion}, Lemma~2, Lemma~3, Theorem~5~(a), and their proofs.) It is based
 on the evolution of the global averages $ \overline{x}(k)=\frac{1}{N}\sum_{i=1}^N x_i(k)$, and $\overline{y}(k)=\frac{1}{N}\sum_{i=1}^N y_i(k)$. Let:
 \begin{align}
 \label{eqn-inexact-oracle-info}
 \widehat{f}_{k-1} &:= \sum_{i=1}^N \left( f_i(y_i(k-1))\phantom{\nabla f_i(y_i(k-1))^\top}\right.\\
 &
 \nonumber
 \left.+\nabla f_i(y_i(k-1))^\top (\overline{y}(k-1)-y_i(k-1))  \right) \\
 \widehat{g}_{k-1} &:= \sum_{i=1}^N \nabla f_i \left( y_i(k-1) \right) \nonumber \\
 L_{k-1} &:= \frac{N}{\alpha_{k-1}} = \frac{N k}{c} \geq 2 \,N\, L\, k \nonumber \\
 \delta_{k-1} &:= L \,\|\widetilde{y}(k-1)\|^2. \nonumber
 \end{align}
Then, it is easy to show that $\overline{x}(k),\,\overline{y}(k)$ evolve as:
\begin{eqnarray}
\label{eqn-CM-evolution-1}
\overline{x}(k) &=& \overline{y}(k-1) - \frac{\widehat{g}_{k-1}}{L_{k-1}}\\
\label{eqn-CM-evolution-2}
\overline{y}(k) &=& \left( 1+\beta_{k-1}\right) \overline{x}(k) - \beta_{k-1} \overline{x}(k-1),
\end{eqnarray}
$k=1,2,\cdots,$ with $\overline{x}(0)=\overline{y}(0)$. As shown in~\cite{arxivVersion},
$\left( \widehat{f}_{k-1},\widehat{g}_{k-1}\right)$
is a $(L_{k-1},\delta_{k-1})$ inexact oracle, i.e., it holds that for all points $x \in {\mathbb R}^d$:
\begin{align}
\label{eqn-inexact-oracle-relations}
&f(x)+\widehat{g}_{k-1}^\top(x-\overline{y}(k-1))
\leq f(x) \leq \widehat{f}_{k-1}\\
&
\nonumber
 + \widehat{g}_{k-1}^\top \left( x - \overline{y}(k-1)\right)
+ \frac{L_{k-1}}{2}\|x-\overline{y}(k-1)\|^2+\delta_{k-1},
\end{align}
%
%
%From~\eqref{eqn-modified-D-NG-vector-form}--\eqref{eqn-modified-D-NG-vector-form-2},
%it can be shown that $\overline{x}(k)$ and $\overline{y}(k)$
% evolve according to \eqref{eqn-lemma-updates-2} and \eqref{eqn-lemma-updates-3},
% with $L_{k-1}:=N/\alpha_{k-1}$ and $\widehat{g}_{k-1}
% := \sum_{i=1}^N \nabla f_i(y_i(k-1))$.
% Further, denote by
%\[
% \widehat{f}_{k-1}:=\sum_{i=1}^N \left( f_i(y_i(k-1))+\nabla f_i(y_i(k-1))^\top (\overline{y}(k-1)-y_i(k-1))  \right).
% \]
% it can be shown that
% $\left(  \widehat{f}_{k-1}, \widehat{g}_{k-1}\right)$
% is a $(L_{k-1},\delta_{k-1})$ inexact oracle,
% where $\delta_{k-1}:=L\|\widetilde{y}(k-1)\|^2.$
%
From~\eqref{eqn-inexact-oracle-info},
 $\widehat{f}_{k-1},\,\widehat{g}_{k-1},$ and $\widehat{\delta}_{k-1}$
 are functions (solely) of~$y(k-1)$.
 Inequalities~\eqref{eqn-inexact-oracle-relations} hold for
 any random realization of $y(k-1)$ and any $x \in {\mathbb R}^d$.
We apply now Lemma~2 in~\cite{arxivVersion}, with $\delta_{k-1}$ as in~\eqref{eqn-inexact-oracle-info}. Get:
 %
 %
% {
%\allowdisplaybreaks
\begin{align}
 \label{eqn-star-star-proof-random}
 &(k+1)^2 \left( f(\overline{x}(k))-f^\star\right)
 +\frac{2 N k}{c}\|\overline{v}(k)-x^\star\|^2 \\
 &
 \nonumber
 \leq\!
 (k^2-1)\!\left( f(\overline{x}(k-1))-f^\star\right)
 \!+\!\frac{2 N k}{c}\|\overline{v}(k-1)-x^\star\|^2\\
 &
 \nonumber
  + (k+1)^2 L\|\widetilde{y}(k-1)\|^2, \nonumber
 \end{align}
%}
%\underline{\underline{\underline{}}}%
%
\noindent
where $\overline{v}(k)=\left( \overline{y}(k) - (1-\theta_{k})\overline{x}(k)\right)/ {\theta_{k}}.$
Dividing~\eqref{eqn-star-star-proof-random} by $k$ and unwinding the resulting inequality, get:
{\allowdisplaybreaks
 \begin{align}
 \label{eqn-star-star-proof-random-trian}
 \frac{1}{N}\left( f(\overline{x}(k))-f^\star\right)
 &\leq
 \frac{2}{k\,c}\|\overline{x}(0)-x^\star\|^2 \\
 &
 \nonumber
 +
 \frac{L}{N\,k}\sum_{t=1}^k \frac{(t+1)^2}{t} \|\widetilde{y}(t-1)\|^2,
 \end{align}
 }
Next, using Assumption~\ref{assumption-bdd-gradients}, obtain, $\forall i$:
\begin{equation}
 \label{eqn-star-star-proof-random-trian-trian}
 %\hspace{-.3cm}
 \frac{1}{N}\!\left( f(x_i(k))\!\!-\!\!f^\star\right) \!\! \leq \!\!
 \frac{1}{N}\!\left( f(\overline{x}(k))\!\!-\!\!f^\star\!\right)\!\! + \!\!\frac{G} {\sqrt{N}} \|\widetilde{x}(k)\|\!.
 \hspace{-.3cm}
\end{equation}
The proof is completed after combining~\eqref{eqn-star-star-proof-random-trian} and~\eqref{eqn-star-star-proof-random-trian-trian},
taking expectation, and using in Theorem~\ref{theorem-consensus-D-NG-modified} the bounds
$\mathbb E \left[ \|\widetilde{x}(k)\| \right] \!\!\leq\!\! \mathbb E \left[ \|\widetilde{z}(k)\| \right]$
 and $\mathbb E \left[ \|\widetilde{y}(k)\|^2 \right]\!\! \leq \!\! \mathbb E \left[ \|\widetilde{z}(k)\|^2 \right]$.

\section{Algorithm $\mathrm m$D--NC}
\label{section-mD-NC}
We present mD--NC. Subsection~\ref{subsection-D-NC-model-algorithm}
 defines additional random matrices needed for representation of mD--NC
 and presents mD--NC. Subsection~\ref{subsection-conv-rate-statement-mD-NC}
 states our result on its convergence rate.
\subsection{Model and algorithm}
\label{subsection-D-NC-model-algorithm}
%
%
%
%

%\subsection{Model and algorithm}
%\label{subsection-model-algorithm}
%
%
%
%This subsection presents our mD--NC algorithm.
We consider a sequence of i.i.d. random matrices that obey Assumptions~\ref{assumption-random-weight-matrices} and~\ref{assumption-connected-random}. We index these matrices with two-indices  since D--NC operates in two time scales--an inner loop,
 indexed by~$s$ with~$\tau_k$ iterations, and an outer loop indexed by~$k$, where:
  \begin{equation}
  \label{eqn-tau-k-modified-D-NC}
  \tau_k = \left\lceil  \frac{3\,\log k + \log N}{-\log {\overline{\mu}}}  \right\rceil.
  \end{equation}
  For static networks, the term $\log N$ can be dropped.
  At each inner iteration,
  nodes utilize one communication round--each node broadcasts a $2 d \times 1$
   vector to all its neighbors.
 We denote by $W(k,s)$ the random weight
 matrix that corresponds to the communication round at the $s$-th inner
 iteration and $k$-th outer iteration.
 The matrices $W(k,s)$ are ordered lexicographically
 as $W(k=1,s=1),W(k=1,s=2),\cdots,W(k=1,s=\tau_1),\cdots,W(k=2,s=1),\cdots$ This sequence obeys Assumptions~\ref{assumption-random-weight-matrices} and~\ref{assumption-connected-random}.

It will be useful to define the products
of the weight matrices $W(k,s)$ over each outer iteration $k$:
\begin{equation}
\label{eqn-cal-W-k-D-NC-modified}
\mathcal{W}(k):=\Pi_{s=0}^{\tau_k-1} W(k,\tau_k-s).
\end{equation}
Clearly, $\{\mathcal{W}(k)\}_{k=1}^{\infty}$
 is a sequence of independent (but not identically distributed) matrices.
 We also define~$\widetilde{\mathcal W}(k):=\mathcal W(k)-J,$ and, for $t=0,1,\cdots,k-1$:
\begin{align}
\label{eqn-psi-k-t}
\widetilde{\Psi}(k,t):=\mathcal W(k) \mathcal W(k-1)\cdots\mathcal{W}(t+1).
\end{align}
%
%for $t=0,1,\cdots ,k-1$.

The Lemma below, proved in Appendix~\ref{app:A}, follows from Assumptions~\ref{assumption-random-weight-matrices} and~\ref{assumption-connected-random},
the independence of the matrices $\mathcal W(k)$,
the value of $\tau_k$ in~\eqref{eqn-tau-k-modified-D-NC}, and Lemma~\ref{lemma-decay-moments}.
\begin{lemma}
\label{lemma-decay-moments-D-NC-modified}
Let Assumptions~\ref{assumption-random-weight-matrices} and~\ref{assumption-connected-random} hold. Then,
for all $k=1,2,\cdots $, for all $s,t \in \{0,1,\cdots ,k-1\}$:
{
\allowdisplaybreaks
\begin{align}
\label{eqn-lemma-psi-norm-1}
%\hspace{-.4cm}
\mathbb E \left[ \left\| \widetilde{\mathcal W}(k)  \right\|^2 \right]\!\! &\leq\!\! \frac{1}{k^6}\\
\label{eqn-lemma-psi-norm-2}
%\hspace{-.4cm}
\mathbb E \left[ \|\widetilde{\Psi}(k,t)\|\right] \!\! &\leq\!\! \frac{1}{k^3(k-1)^3\cdots (t+1)^3}\\
\label{eqn-lemma-psi-norm-3}
%\hspace{-.4cm}
\mathbb E\! \left[ \!\|\widetilde{\Psi}(k,t)^\top \!\widetilde{\Psi}(k,t) \|\!\right] \!\! &\leq\!\!\left(\!\!\frac{1}{k^3(k-1)^3\cdots (t+1)^3}\!\!\right)^2\\
\label{eqn-lemma-psi-norm-4}
%
%\hspace{-.4cm}
\mathbb E\! \left[ \!\|\widetilde{\Psi}(k,s)^\top\! \widetilde{\Psi}(k,t) \|\!\right]
 \!\! &\leq\!\! \left(\!\!\frac{1}{k^3(k-1)^3\cdots (t+1)^3}
 \!\!\right)\\
% \hspace{-.4cm}
&
 \nonumber
 \phantom{\!\! \leq\!\!\ \ }
 \left(\!\!\frac{1}{k^3(k-1)^3\cdots (s+1)^3} \!\!\right).
\end{align}
}
%
%where ${\overline{\mu}}$ is defined in Lemma~\ref{lemma-decay-moments}.
\end{lemma}

\textbf{The mD--NC algorithm}.
mD--NC, in Algorithm~1, uses constant step-size $\alpha \leq 1/(2L)$.
 Each node $i$ maintains over (outer iterations) $k$ the solution estimate
  $x_i(k)$ and an auxiliary variable~$y_i(k)$.
   Recall ${\overline{\mu}}$ in Lemma~\ref{lemma-decay-moments}.
\begin{algorithm}
\label{alg-D-NC-modified}
\caption{mD--NC}
\begin{algorithmic}[1]
{\small{
    \STATE Initialization: Node $i$ sets $x_i(0)\!=\!y_i(0)\! \in {\mathbb R}^d$; and $k=1.$
    %\STATE
        %\STATE Node $i$ calculates~$\nabla f_i(y_i(k-1))$.
        %\STATE (First consensus) Nodes run~\eqref{eqn-consensus}
%        for $s=1,2,\cdots \tau_{k-1}$, with $z_i(s=0,k-1)=\nabla f_i(y_i(k-1))$, so
%        that node $i$ obtains~$g_i(k-1):=z_i(s=\tau_{k-1},k-1)$.
        \STATE Node $i$ calculates
        $
        x_i^{(a)}(k) = y_i(k-1) - \alpha \nabla f_i(y_i(k-1)).
        $
        \STATE (Consensus) Nodes run average consensus on
         $\chi_i(s,k)$, initialized by~$\chi_i(s=0,k)=\left(x_i^{(a)}(k)^\top,x_i(k-1)^\top\right)^\top$:
        \begin{align*}
        %\label{eqn-consensus-D-NC-modified}
        \chi_i(s,k)\!\! = \!\!\sum_{j \in O_i(k)}\!\! W_{ij}(k,s)\chi_j(s-1,k),s=1,2,\cdots,\tau_k,
        \end{align*}
        with $\tau_k$ in~\eqref{eqn-tau-k-modified-D-NC},
        and set~$x_i(k):=\left[\chi_i(s=\tau_k,k)\right]_{1:d}$
         and $x_i^{(b)}(k-1):=\left[\chi_i(s=\tau_k,k)\right]_{d+1:2\,d}$.
         (Here $[a]_{l:m}$ is a selection of $l$-th, $l+1$-th,
         $\cdots$ , $m$-th entries of vector $a$.)
        \STATE Node $i$ calculates
        $
        y_i(k) \!\!=\!\! (1+\beta_{k-1}) x_i(k)\!\! - \!\!\beta_{k-1}\,x_i^{(b)} \!(k-1).
        $
        %\STATE (Second consensus) Nodes run average consensus initialized by~$y_i^{(c)}(s=0,k)=y_i^{(a)}(k)$:
%        \begin{eqnarray}
%        \label{eqn-consensus-2}
%        y_i^{(c)}(s,k) = \sum_{j \in O_i} W_{ij} y_j^{(c)}(s-1,k),\:\:s=1,2,\cdots ,\tau_y(k),\:\:\:\:\:\:\:
%        %\label{eqn-consensus}
%        \tau_y(k) = \left\lceil \frac{\log 3}{-\log {\overline{\mu}}(W)} + \frac{2 \log k }{-\log {\overline{\mu}}(W)} \right\rceil,
%        \end{eqnarray}
%        and set~$y_i(k):=y_i^{(c)}(s=\tau_y(k),k)$.
        \STATE Set $k \mapsto k+1$ and go to step 2.
        }
        }
\end{algorithmic}
%\label{algorithm_SNG}
\end{algorithm}
Step~3 has $\tau_k$ communication rounds at
 outer iteration~$k$.
Nodes know $L$, $\overline{\mu}$, and $N$.
 Section~\ref{section-discussion} relaxes this.

\textbf{mD-NC in vector form}. Let the matrices~$\mathcal{W}(k)$ in~\eqref{eqn-cal-W-k-D-NC-modified}.
Use the compact notation in mD--NG for
 $x(k)$,
 $y(k)$, and $F:{\mathbb R}^{N d} \rightarrow {\mathbb R}^N$. Then for $k=1,2,\cdots$
\begin{align}
\label{eqn-modified-D-NC-vector-form}
x(k) &\!=\! \left(\!\mathcal{W}(k) \!\otimes \!I\!\right)\! \left[\! y(k-1)\! -\! \alpha \nabla F(y(k-1)\! \right] \\
\label{eqn-modified-D-NC-vector-form-2}
y(k) &\!=\! (1\!+\!\beta_{k-1})x(k) \!-\!\beta_{k-1} \!\left(\mathcal{W}(k)\! \otimes\! I\right)\!x(k\!-\!1),
\end{align}
  with $x(0)=y(0) \in {\mathbb R}^{N d}.$   Note the formal similarity
  with mD--NG in~\eqref{eqn-modified-D-NG-vector-form}--\eqref{eqn-modified-D-NG-vector-form-2}, except that
   $W(k)$ is replaced
   by $\mathcal{W}(k)$, and the diminishing
    step-size $\alpha_k=c/(k+1)$
     is replaced with the constant step-size $\alpha_k=\alpha.$

%We assumed that nodes know beforehand the global parameter $L$,
% with m--NG, and $L,\overline{\mu}$, and $N$, with mD--NC.
%There may be scenarios where such knowledge is not available;
% the two methods in such scenarios can be adapted to operate without
% the knowledge of global parameters; see the end of Subsection~\ref{subsection-rates-statements}.
%
%
%
%
%
%
%

\subsection{Convergence rate}
\label{subsection-conv-rate-statement-mD-NC}
 Define, like for mD--NG,
the disagreements $\widetilde{x}_i(k)$, $\widetilde{y}_i(k)$,
$\widetilde{x}(k)$, and $\widetilde{y}(k),$
 and  $\widetilde{z}(k):=\left(  \widetilde{y}(k)^\top,\widetilde{x}(k)^\top \right)^\top.$
% We have Theorem~\ref{theorem-consensus-D-NC-modified}  on the disagreement bounds.
 %
 %
 %
\begin{theorem}
\label{theorem-consensus-D-NC-modified}
Consider mD--NC in Algorithm~1 under Assumptions~\ref{assumption-random-weight-matrices}--\ref{assumption-bdd-gradients}.
Then, for all $k=1,2,\cdots $
{
\allowdisplaybreaks
\begin{align}
\label{eqn-theorem-norm-z-k-D-NC-mod}
{\mathbb E}\left[ \|\widetilde{z}(k)\|  \right] &\leq \frac{50\,\alpha\,N^{1/2}\,G}{k^2}\\
\label{eqn-theorem-norm-z-k-squared-D-NC-mod}
{\mathbb E}\left[ \|\widetilde{z}(k)\|^2  \right] &\leq  \frac{50^2\,\alpha^2\,N G^2} {k^4}.
\end{align}
}
%
%
%where ${\overline{\mu}}$
\end{theorem}

\begin{theorem}
\label{theorem-conv-rate-D-NC-modified}
Consider mD--NC in Algorithm~1 under Assumptions~\ref{assumption-random-weight-matrices}--\ref{assumption-bdd-gradients}.
Let $\|\overline{x}(0)-{{x^\star}}\| \leq R$, $R\geq0$. Then,
 after~$\mathcal{K}$ communication rounds (after~$k$ outer iterations)
 \begin{align*}
 %\label{eqn-comunication-rounds-modif}
 \mathcal{K} = \sum_{t=1}^{k}\tau_t \leq \frac{3}{-\log {\overline{\mu}}} \left[ \,(k+1)\log (N(k+1)) \,\right],
 \end{align*}
  i.e., $\mathcal{K}\!\sim\! O\left(k \log k\right)$, we have, at any node $i$, $k=1,2,\cdots$
\begin{align*}
%\label{eqn-theorem-D-NC-modif}
 \frac{\mathbb E \left[ f(x_i(k)) - f^\star \right]}{N} \leq
\frac{1}{k^2} \left( \frac{2}{\alpha}  R^2
+ 11 \,\alpha^2 L G^2 + \alpha G^2 \right),
\end{align*}
\end{theorem}

\textbf{Remark}. Theorem~\ref{theorem-conv-rate-D-NC-modified} implies $ \frac{\mathbb E \left[ f(x_i(k)) - f^\star \right]}{N}$ with mD--NC converges at rate $O(1/\mathcal K^{2-\xi})$ in the number of communications $\mathcal K$.
%This is proved analogously to Theorem~9 in~\cite{arxivVersion}.

\section{Proofs of Theorems~\ref{theorem-consensus-D-NC-modified} and~\ref{theorem-conv-rate-D-NC-modified}}
\label{section-proofs-mD-NC}
We now prove the convergence rate results for mD--NC. Subsection~\ref{subsection-proof-consensus-mD-NC-random-net}
proves Theorem~\ref{theorem-consensus-D-NC-modified} and Subsection~\ref{subsection-proof-mD-NC-rate-convergence}
 proves Theorem~\ref{theorem-conv-rate-D-NC-modified}.
 %Finally, Subsection~\ref{subsection-discussion}
% discusses some implications of our results and extends the convergence rate results (of both mD--NG and mD--NC) to scenarios when the global parameters ($L$, $N$, and $\overline{\mu}$) are not available beforehand.
%%
%
\subsection{Proof of Theorem~\ref{theorem-consensus-D-NC-modified}}
\label{subsection-proof-consensus-mD-NC-random-net}
For simplicity, we prove for $d=1$, but the proof extends to generic~$d>1$.
Similarly to Theorem~\ref{theorem-consensus-D-NG-modified},
we proceed in three steps. In Step~1,
we derive the dynamics for the
disagreement~$\widetilde{z}(k)=\left( \widetilde{y}(k)^\top,\,\widetilde{x}(k)^\top \right)$.
 In Step~2, we unwind the disagreement equation and express
 $\widetilde{z}(k)$ in terms of the $\widetilde{\Psi}(k,t)$'s in~\eqref{eqn-psi-k-t}
  and $\mathcal{B}(k,t)$ in~\eqref{eqn-B-products}.
  Step~3 finalizes the proof using bounds previously established
  on the norms of $\widetilde{\Psi}(k,t)$ and~$\mathcal{B}(k,t).$

\textbf{Step~1. Disagreement dynamics}.
We write the dynamic equation for $\widetilde{z}(k)$.
Recall $B(k)$ in~\eqref{eqn-B-k-definition}.
Multiplying~\eqref{eqn-modified-D-NC-vector-form}--\eqref{eqn-modified-D-NC-vector-form-2} from the left by $(I-J)$, and using
$(I-J){\mathcal W}(k)=\widetilde{\mathcal W}(k)(I-J)$, obtain for $k=1,2,\cdots$
\begin{align}
\label{eqn-recursion-D-NC-modified}
\widetilde{z}(k) = \left( B(k) \otimes \widetilde{\mathcal W}(k)\right) \,\left( \widetilde{z}(k-1) + u^\prime(k-1)\right) ,
\end{align}
and $\widetilde{z}(0)=0$,
where
{\small{
\begin{eqnarray}
\label{eqn-u-k-prime}
u^\prime(k-1) = -\left[ \begin{array}{cc} \alpha_{k-1}\,\nabla F(y(k-1))\\
0 \end{array} \right].
\end{eqnarray}
}}

\textbf{Step~2: Unwinding the recursion~\eqref{eqn-recursion-D-NC-modified}}.
 Recall~$\mathcal{B}(k,t)$ in~\eqref{eqn-B-products}.
Unwinding~\eqref{eqn-recursion-D-NC-modified} and using
$(A \otimes B)(C \otimes D)=(AC) \otimes (BD)$, obtain for $k=1,2,\cdots$
\begin{align}
\label{eqn-unwinded-D-NC-modified}
\widetilde{z}(k)\!\! = \!\!\sum_{t=0}^{k-1}\!\! \left(\! \mathcal{B}(k,k-t-2)\!B(t+1) \!\!\otimes \!\!\widetilde{\Psi}(k,t) \!\right)\!\!u^\prime(t),
\end{align}
 The quantities
$u^\prime(t)$ and $\widetilde{\Psi}(k,t)$ in~\eqref{eqn-unwinded-D-NC-modified}
are random, while the $\mathcal{B}(k,k-t-2)$'s are deterministic.

\textbf{Step~3: Finalizing the proof}. Consider $u^\prime(t)$ in~\eqref{eqn-u-k-prime}.
By Assumption~\ref{assumption-bdd-gradients},
$\|\nabla F(y(t))\| \leq \sqrt{N} G$.
From this, obtain $\|u^\prime(t)\| \leq \alpha\,\sqrt{N}\,G$, for any
random realization of~$u^\prime(t)$.
Using this, Lemma~\ref{lemma-upper-bound-on-norm-B-k-t},
the sub-multiplicative and sub-additive properties of norms,
and $\|B(t+1)\|\leq 3$, $\forall t$, get from~\eqref{eqn-unwinded-D-NC-modified}:
\begin{align*}
\|\widetilde{z}(k)\|& \leq 3\left( 8\alpha \sqrt{N}G \right)
\sum_{t=0}^{k-1} \|\widetilde{\Psi}(k,t)\|  \frac{(k-t-1)(t+1)}{k}\\
&
+
3\left( 5  \alpha  \sqrt{N}\,G \right)
\sum_{t=0}^{k-1} \|\widetilde{\Psi}(k,t)\|. %\nonumber
\end{align*}
Taking expectation, using
$\frac{(k-t-1)(t+1)}{k} \leq t+1$ and Lemma~\ref{lemma-decay-moments-D-NC-modified}, \eqref{eqn-theorem-norm-z-k-D-NC-mod} follows from
\begin{align*}
&\mathbb E\! \left[\! \|\widetilde{z}(k)\| \!\right]\!\! \leq\!\! 3\!\!\left( \!8 \alpha\sqrt{N}G \!\right)\!\!
\sum_{t=0}^{k-1}\!\! \frac{1}{k^3(k\!-\!1)^3\cdots (t\!+\!2)^3(t\!+\!1)^2} \\
&
+
3\left( 5 \,\alpha\, \sqrt{N}\,G \right)\,
\sum_{t=0}^{k-1} \frac{1}{k^3(k-1)^3\cdots (t+2)^3(t+1)^3} \\
&\leq
 3\left( 8 \,\alpha\, \sqrt{N}\,G \right)\,
\frac{1}{k^2} \\
&+
3\left( 5 \,\alpha\, \sqrt{N}\,G \right)\,
 \frac{1}{k^2} \leq \frac{50\,\alpha \,\sqrt{N}\,G}{k^2}.
\end{align*}

We prove~\eqref{eqn-theorem-norm-z-k-squared-D-NC-mod}.
Consider $\|\widetilde{z}(k)\|^2.$
Get from~\eqref{eqn-unwinded-D-NC-modified}:
\begin{align*}
&\|\widetilde{z}(k)\|^2
\!\! =\!\!\!
 \sum_{t=0}^{k-1}\!\sum_{s=0}^{k-1}\!\!
 u^\prime(t)^\top\!\!\!\left( \!\!B(t\!+\!1)^\top\!\!\mathcal{B}(k,k\!-\!t\!-\!2)^\top\!\! \!\otimes \!\!\widetilde{\Psi}(k,t)^\top\!\! \right)\\
 &
 \times\,\left(
 \mathcal{B}(k,k-t-2)B(t+1) \otimes \widetilde{\Psi}(k,s)\right)\,u^\prime(s) \\
 &=
 \sum_{t=0}^{k-1}\sum_{s=0}^{k-1}
 u^\prime(t)^\top \,( B(t+1)^\top\mathcal{B}(k,k-t-2)^\top \\
 &\times \mathcal{B}(k,k\!-\!s\!-\!2)B(t\!+\!1) )\otimes\left(
 \widetilde{\Psi}(k,t)^\top  \widetilde{\Psi}(k,s)\right)\,u^\prime(s),
\end{align*}
where the last inequality uses $(A \otimes B)(C \otimes D)
= (A C) \otimes (B D)$.
 By the sub-additive
 and sub-multiplicative properties of norms, and
 $\|B(t+1)\|\leq 3$, $\forall t$, obtain:
\begin{align*}
\nonumber
\|\widetilde{z}(k)\|^2
 &\leq\!
9\!\sum_{t=0}^{k-1}\sum_{s=0}^{k-1}\!
 \left\| \mathcal{B}(k,k-t-2)\right\|\! \left\|\mathcal{B}(k,k-s-2) \right\|\\
 &
%\hspace{1.25cm}
 \left\| \widetilde{\Psi}(k,t)^\top \, \widetilde{\Psi}(k,s)\right\|
 \,\|u^\prime(t)\|\,\|u^\prime(s)\| \\
 &\leq
9 \sum_{t=0}^{k-1}\sum_{s=0}^{k-1}
 \left( 8 (t+1) + 5\right)\left( 8(s+1) + 5 \right)  \\
 &
% \hspace{1.25cm}
 \left\|\widetilde{\Psi}(k,t)^\top \widetilde{\Psi}(k,s)\right\|
 {{\alpha }^2\,N G^2}, \nonumber
\end{align*}
where the last inequality uses $(k-s-1)(s+1)/k \leq s+1$, Lemma~\ref{lemma-upper-bound-on-norm-B-k-t}
 and $\|u(t)\| \leq \left(\alpha \,\sqrt{N} G\right)/(t+1).$
 Taking expectation and applying Lemma~\ref{lemma-decay-moments-D-NC-modified}, we obtain:
 {
 \allowdisplaybreaks
\begin{align*}
 &\mathbb E \left[ \|\widetilde{z}(k)\|^2    \right] \leq  9\left({\alpha}^2 N G^2\right)
 \sum_{t=0}^{k-1}\sum_{s=0}^{k-1}
 \left( 8(t+1) + 5\right)\\
 &
 \left( 8(s+1)+ 5 \right)
 \left( \frac{1}{k^3\cdots (t+1)^3} \right) \,\left(  \frac{1}{k^3\cdots (s+1)^3}  \right) \\
 &
 =
 9\left( {\alpha}^2\,N G^2\right)
 \left( \sum_{t=0}^{k-1}  \left( 8(t+1) + 5\right) \frac{1}{k^3\cdots (t+1)^3}  \right)^2 \\
 &
 \leq
 \frac{ 50^2\,{\alpha}^2\,N G^2} {k^4}.
 \end{align*}
 }
%
%
%where the last inequality applies Lemma~\ref{lemma-bounds-on-sums}.
Thus, the bound in~\eqref{eqn-theorem-norm-z-k-squared-D-NC-mod} and Theorem~\ref{theorem-consensus-D-NC-modified} is proved.

\subsection{Proof outline of Theorem~\ref{theorem-conv-rate-D-NC-modified}}
\label{subsection-proof-mD-NC-rate-convergence}
 We outline the proof since similar to Theorem~8 in~\cite{arxivVersion}~(version v2).
 Consider the global averages $\overline{x}(k)$
  and $\overline{y}(k)$ as in mD--NG.
  Then, $\overline{x}(k)$ and $\overline{y}(k)$
  follow~\eqref{eqn-CM-evolution-1}--\eqref{eqn-CM-evolution-2},
   with $L_{k-1}:=N/\alpha$ and $\widehat{g}_{k-1}$
    as in~\eqref{eqn-inexact-oracle-info}.
    Inequalities~\eqref{eqn-inexact-oracle-relations} hold
    with $L_{k-1}:=N/\alpha$ and $\widehat{f}_{k-1}$ and $\widehat{g}_{k-1}$
     as in~\eqref{eqn-inexact-oracle-info}.
     Applying Lemma~2 in~\cite{arxivVersion} gives:
    \begin{align*}
    \frac{1}{N}\left( f(\overline{x}(k)) -f^\star \right)
    &
    \leq
    \frac{2}{\alpha\,k^2}\|\overline{x}(0)-x^\star\|^2\\
    &
    +\frac{L}{N\,k^2}\sum_{t=1}^k \|\widetilde{y}(t-1)\|^2(t+1)^2.
    \end{align*}
    (Compare the last equation with~\eqref{eqn-star-star-proof-random-trian}.)
    The remainder of the proof proceeds analogously
    to that of Theorem~\ref{theorem-convergence-rate-D-NG-modified}.
\section{Discussion and extensions}
\label{section-discussion}
We discuss extensions and corollaries: 1) relax the prior knowledge on $L,\overline{\mu}$, and $N$ for both mD--NG and mD--NC; 2) establish
rates in the convergence in probability of mD--NG and mD--NC; 3) show
almost sure convergence with mD--NC; and 4) establish a convergence rate in the second moment
with both methods.

\textbf{Relaxing knowledge of $L,\overline{\mu}$, and $N$}. %With mD--NG and mD--NC, we required
%a priori knowledge of the parameters $L,\overline{\mu}$, and $N$.
% Actually, the convergence rates (the same with D--NG and a deteriorated one with D--NC)
%are guaranteed if this knowledge is not available, as we explain next.
%
 mD--NG requires only knowledge of $L$ to set the step-size $\alpha_k=c/(k+1)$, $c \leq 1/(2 L)$.
  We demonstrate
 that the rate $O(\log k/k)$ (with a deteriorated constant) still holds if nodes use arbitrary $c>0$.
   Initialize all nodes to $x_i(0)=y_i(0)=0$, suppose that $c>1/(2 L)$, and let $k^\prime = 2 \,c \,L$.
  Applying Lemma 2 in~\cite{arxivVersion}, as in the proof of Theorem~5~(b) in~\cite{arxivVersion},
  for all $k > k^\prime$, surely:
 \begin{align}
 \label{eqn-star-star-proof-random-2}
 &\frac{(k+1)^2-1}{k+1} \left( f(\overline{x}(k))-f^\star\right) \\
 &
 \nonumber
 \leq
 k^\prime \!\left(\! f(\overline{x}(k^\prime\!-\!1))\!-\!f^\star\!\right)
 \!+\!\frac{2 N }{c}\!\left(2\|\overline{v}(k^\prime-1)\|^2\!+\! 2\|x^\star\|^2 \right)\\
 &
 \nonumber
 +
 \sum_{t=1}^k \frac{(t+1)^2}{t} L\|\widetilde{y}(t-1)\|^2.
 \end{align}
  Further, from the proof of Theorem~5~(b) in~\cite{arxivVersion}, surely:
  \begin{equation}
  \label{eqn-modificat}
  \|\overline{v}(k^\prime-1)\|^2 \leq (2k^\prime+1)^2 \left(3^{k^\prime}\right)^2\,2\,c\,G.
  \end{equation}
  Finally, Theorem~\ref{theorem-consensus-D-NG-modified}
  holds unchanged for $c>1/(2L)$. Thus
  $\sum_{t=1}^k \frac{(t+1)^2}{t} L \mathbb E\left[\|\widetilde{y}(t-1)\|^2\right]=O(\log k)$.
 Multiplying~\eqref{eqn-star-star-proof-random-2} by $\frac{k+1}{(k+1)^2-1}$, taking expectation
 on the resulting inequality, and
 applying Theorem~\ref{theorem-consensus-D-NG-modified}, obtain desired
 $O(\log k/k)$ rate.
%  The algorithm's progress per iteration (see~\eqref{eqn-star-star-proof-random})
%   holds whenever $\frac{N}{\alpha_{k-1}}=\frac{N k}{c} \geq 2 N L$.
%   Thus, the progress equation~\eqref{eqn-star-star-proof-random} holds for any $k \geq 2\,c\,L$,
%   and the rate $O(\log k/k)$ is achieved, with a deteriorated constant.
%   This can be seen from the proof of Theorem~5~(b) in~\cite{arxivVersion}.
%   Namely, equation~(73) in~\cite{arxivVersion} continues
%   to hold here, surely. Then, we obtain the $O(\log k/k)$ rate
%   by taking expectation in equation~(69) of~\cite{arxivVersion}.

 mD--NC uses the constant step-size $\alpha \leq 1/(2 L)$
 and $\tau_k$ in~\eqref{eqn-tau-k-modified-D-NC}.
 To avoid the use of $L,\overline{\mu}$, and $N$, we set in mD-NC:
1) a diminishing step-size $\alpha_k=1/k^p$, $p\in (0,1]$;
and 2) $\tau_k=k$ (as suggested in~\cite{AnnieChen}).
 We show the adapted mD--NC achieves rate $O(1/k^{2-p})$.
Let $k^{\prime \prime} = (2 L)^{1/p}$. Then,
 by Lemma~2 in~\cite{arxivVersion}, $\forall k \geq k^{\prime \prime}$, surely:
 \begin{align}
 \label{eqn-star-star-proof-random-3}
 &
 \frac{(k+1)^2-1}{(k+1)^p} \left( f(\overline{x}(k))-f^\star\right) \\
 &
 \nonumber
\hspace{2cm}
  \leq
 (k^\prime)^{2-p} \left( f(\overline{x}(k^\prime-1))-f^\star\right)\\
 &
 +\!\! 2 N\!\!\left(2\|\overline{v}(k^\prime\!-\!1)\|^2\!\!+\! 2\|x^\star\|^2\right)\!\!+\!\!
 \sum_{t=1}^k \!\!\frac{(t+1)^2}{t^p}\! L\|\widetilde{y}(t\!-\!1)\|^2\!. \nonumber
 \end{align}
Further, \eqref{eqn-modificat} holds here as well (surely.)
 Modify the argument on the sum in~\eqref{eqn-star-star-proof-random-3}.
By Lemma~\ref{lemma-decay-moments-D-NC-modified} and $\tau_k=k$, we have:
$\mathbb E \left[\| \widetilde{\mathcal W}(k)\|^2 \right] \leq N^2 \overline{\mu}^{2\,k}$.
From this, $\forall k \geq k^{\prime \prime \prime}:=
\left( \frac{6(\log N+1)}{-\log \overline{\mu}}\right)^2  $: $\mathbb E \left[\| \widetilde{\mathcal W}(k)\|^2 \right] \leq \frac{1}{k^4}$.
Next, consider
\begin{align*}
\widetilde{\Psi}(k,s)^{\!\!\top}\!\widetilde{\Psi}(k,t)
\!\!=\!\!
\left(\!\widetilde{\mathcal{W}}(k) \!\cdots\! \widetilde{\mathcal{W}}(s+1)\!\!\right)^{\!\!\!\top}\!\!\!\!
\left(\!\widetilde{\mathcal{W}}(k) \!\cdots\! \widetilde{\mathcal{W}}(t+1)\!\!\right)\!\!,
\end{align*}
 for arbitrary $k \geq k^{\prime \prime \prime}$,
 and arbitrary $s,t\! \in \!\{0,1,\!\cdots\!,k\!-\!1\}$. Clearly,
 %\[
 $
 \left\|  \widetilde{\Psi}(k,s)^\top \widetilde{\Psi}(k,t) \right\|
 \!\!\leq\!\! \left\| \widetilde{\mathcal{W}}(k) \right\|^2\!\!\!,
% \]
$
and hence:
\[
\mathbb E \!\left[ \left\|  \widetilde{\Psi}(k,s)^{\!\!\top}\!\! \widetilde{\Psi}(k,t) \right\|\right]
\!\!\leq\!\!
\frac{1}{k^4},\forall s,t\!\! \in\!\! \{0,1,\!\!\cdots\!\! ,k-1\},\!k\!\! \geq \!\!k^{\prime \prime \prime}\!\!.
\]
Now, from step~3 of the proof of Theorem~\ref{theorem-consensus-D-NC-modified}, the above implies: $
\mathbb E \left[ \|\widetilde{y}(k)\|^2 \right]
\leq
\mathbb E \left[ \|\widetilde{z}(k)\|^2 \right]\leq \frac{C}{k^4}$,
for all $k \geq k^{\prime \prime \prime}$, where $C>0$ is independent of $k$.
Hence, we obtain the desired bound on the sum:
\[
 \sum_{t=1}^{\infty} \frac{(t+1)^2}{t^p} L \mathbb E\left[\|\widetilde{y}(t-1)\|^2\right] = O(1).
\]
Using this, \eqref{eqn-modificat}, multiplying \eqref{eqn-star-star-proof-random-3} by
 $\frac{(k+1)^p}{(k+1)^2-1}$, and
taking expectation in~\eqref{eqn-star-star-proof-random-3}, obtains the rate $O(1/k^{2-p})$.
%
% the progress per iteration (similar to~\eqref{eqn-star-star-proof-random})
%  holds whenever $\frac{N}{\alpha_{k-1}} \geq 2 NL$,
%  i.e., for all $k \geq (2 L)^{1/p}$.
%Further, if nodes do not know $\overline{\mu}$ nor $N$ to set
%$\tau_k$ in~\eqref{eqn-tau-k-modified-D-NC}, they can set
%$\tau_k=k,$ as explained in~\cite{AnnieChen}. It can be shown
%that the adapted method (without knowledge of $L,\overline{\mu}$ and $N$)
% achieves the rate $O(1/k^{2-p})$ and $O(1/{\mathcal{\mathcal K}}^{1-p/2}).$

\textbf{Convergence in probability and almost sure convergence}. Through the Markov inequality,
Theorems~\ref{theorem-convergence-rate-D-NG-modified} and~\ref{theorem-conv-rate-D-NC-modified}
 imply, for any $\epsilon >0$, $k \rightarrow \infty,\forall i$:
  \begin{align*}
   & \mathrm{mD-NG}:\,\,\mathbb P \left( k^{1-\xi}\,\left(f(x_i(k))-f^\star \right) > \epsilon \right) \rightarrow 0\:\:
    \\
   & \mathrm{mD-NC}:\,\,
   \mathbb P \left( k^{2-\xi}\,\left(f(x_i(k))-f^\star \right) > \epsilon \right) \rightarrow 0,
     \end{align*}
   where $\xi>0$ is arbitrarily small.
       Furthermore, by the arguments in, e.g., (\cite{jadbabaie_on_consensus}, Subsection~{IV--A}),
    with mD--NC, we have that, $\forall i$,
    $f(x_i(k))-f^\star \rightarrow 0$, almost surely.

\textbf{Convergence rates in the second moment}. Consider a special case of
the random network model $G(k)$ in Assumptions~\ref{assumption-random-weight-matrices} and~\ref{assumption-connected-random}
 that supports a random instantiation of $W(k)$:
 $G(k)=(\mathcal{N},E)$, with $E=\left\{\{i,j\}:\,W_{ij}(k)>0,\:i<j \right\}$.
 We assume $G(k)$ is connected with positive probability.
 This holds with spatio-temporally independent
 link failures, but not with pairwise gossip, where
 one edge occurs at a time,
 hence all realizations of $G(k)$ are disconnected. We establish the bounds on the second moment
of the optimality gaps:
   \begin{align}
   \label{eqn-second-moment-mD-NG}
   & \mathrm{mD-NG}\!:\!\mathbb E\!\! \left[\!\left(f(x_i(k))\!-\!f^\star\! \right)^2\!\right]\!\! =\!\! O\left(\!\!
   \frac{\log^2 k}{k^2}\!\!\right)\!\!,\forall i\!, \\
   \label{eqn-second-moment-mD-NC}
   & \mathrm{mD-NG}\!:\!\mathbb E\!\! \left[\!\left(f(x_i(k))\!-\!f^\star\! \right)^2\!\right]\!\! = \!\!O\left(\frac{1}{k^4}\right),\forall i,
   \end{align}
where \eqref{eqn-second-moment-mD-NC} holds for mD--NC with a modified value of $\tau_k$ (see Appendix~\ref{subsection-appendix-moments}.)
We interpret~\eqref{eqn-second-moment-mD-NG}, while~\eqref{eqn-second-moment-mD-NC} is similar.
 Result~\eqref{eqn-second-moment-mD-NG} shows that, not only the mean of the optimality gap decays
as $O(\log k/k)$ (by Theorem~\ref{theorem-convergence-rate-D-NG-modified}),
but also the standard deviation is $O(\log k /k)$.

\section{Simulation example}
\label{section-simulation-example}
%
%
%A simulation corroborates convergence rates
%of mD--NG and mD--NC with link failures,
%as well as rates faster than the rates of the method in~\cite{nedic_T-AC}.
%We also compare D--NG and mD--NC with the original variants in~\cite{arxivVersion}. Finally,
%we show that the weight optimization in~\cite{weightOpt}
%reduces the convergence constant of mD--NG.
%
 We compare  mD--NG
       and mD--NC, D--NG and D--NC in~\cite{arxivVersion}, and the methods in~\cite{nedic_T-AC,nedic_novo}. We initialize
       all to $x_i(0)=y_i(0)=0$, $\forall i$. We
       generate one sample path (simulation run),
       and estimate the average normalized optimality gap
  $\mathbf{err_f}=\frac{1}{N}\sum_{i=1}^N \frac{f(x_i)-f^\star}{f(0)-f^\star}$
   versus the total number $\mathcal K^\prime$
   of scalar transmissions, across all nodes. We count both the successful and failed
     transmissions. All our plots are in $\log_{10}-\log_{10}$ scales.

\textbf{Setup}. Consider a connected geometric supergraph $\mathcal G=(\mathcal N,E)$
   generated
   by placing $10$~nodes at random on a unit 2D
   square and connecting the nodes whose distance is less than a prescribed radius ($26$ links). We consider random and static networks. With the random graph, nodes fail with probability~$.9$
             For online links $\{i,j\}\in E$,
       the weights
       $W_{ij}(k)=W_{ji}(k)=1/N=1/10$  and $W_{ii}(k)=1-\sum_{j \in O_i(k)-\{i\}}W_{ij}(k)$, $\forall i$.
       The static network has
       the same supergraph $\mathcal G$,
       and, $\forall \{i,j\} \in E$, we set
       $W_{ij}=W_{ji}=1/N$.
        With D--NG and mD--NG,
        the step-size is $\alpha_k=1/(k+1)$,
        while with D--NC and mD--NC,
         $\alpha=1/2$ and with~\cite{nedic_T-AC},
         we use $\alpha_k=1/\sqrt{k}$.
         With random networks,
         for both variants of D--NC, we set
         $\tau_k$ as in~\eqref{eqn-tau-k-modified-D-NC};
         with static networks, we
         use $\tau_k=\left\lceil  \frac{3\,\log k }{-\log {\overline{\mu}}}  \right\rceil$.
          (As indicated in Section~\ref{section-mD-NC}, the $\log N$ term is not needed with static networks.)

We use Huber loss cost functions
arising, e.g., in distributed robust estimation in sensor networks~\cite{Rabbat};
$f_i(x) = \|x-\theta_i\|-1/2$, else,
 $\theta_i \in {\mathbb R}$. The $f_i$'s
 obey Assumptions~\ref{assumption-f-i-s} and~\ref{assumption-bdd-gradients}.
 We set $\theta_i=\pm 4(1+\nu_i)$
 and $\nu_i$   is generated randomly from the uniform distribution on $[-0.1,0.1]$. For $i=1,2,3$, we use the $+$ sign and
 for $i=4,\cdots ,10$ the $-$ sign.

    \textbf{Results: Link failures}. Figure~\ref{figure-1-link-failures}~(top) shows that the convergence rates (slopes) of mD-NG, mD-NC, and D-NC, are better than that of the method in~\cite{nedic_T-AC}. All methods converge, even with severe link failures, while D-NG diverges, see Figure~\ref{figure-1-link-failures}~(second from top plot).

\textbf{Results: Static network}. Figure~\ref{figure-1-link-failures} (second from bottom) shows mD-NG, mD-NC, D-NG, D-NC, and the method in~\cite{nedic_novo} on a static network.  As expected with a static network, D-NG performs slightly better than mD-NG, and both converge faster than~\cite{nedic_novo}. D-NC and mD-NC perform similarly on both static and random networks. The bottom plot in Figure~\ref{figure-1-link-failures} shows mD-NG and D-NG when mD-NG is run with Metropolis weights~$W$, \cite{BoydFusion}, while D-NG, because it requires positive definite weights, is run with positive definite $W^\prime =\frac{1.01}{2}\,I+\frac{0.99}{2} W$. The D-NG performs only marginally better than the mD-NG, which has a larger (worse) $\overline{\mu}$.
    \begin{figure}[thpb]
      \centering
       \includegraphics[height=1.25 in,width=3.18 in]{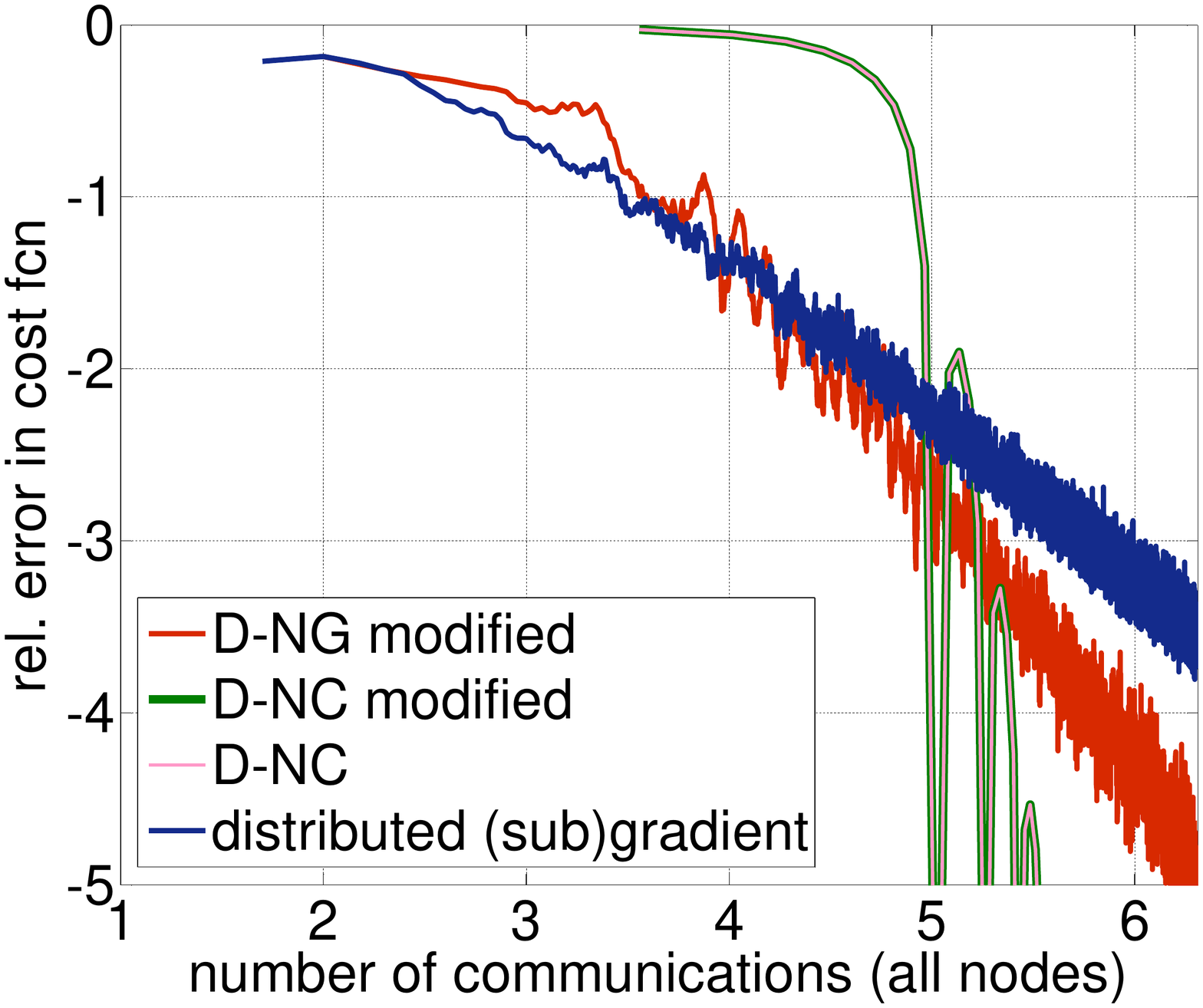}
       \includegraphics[height=1.25 in,width=3.18 in]{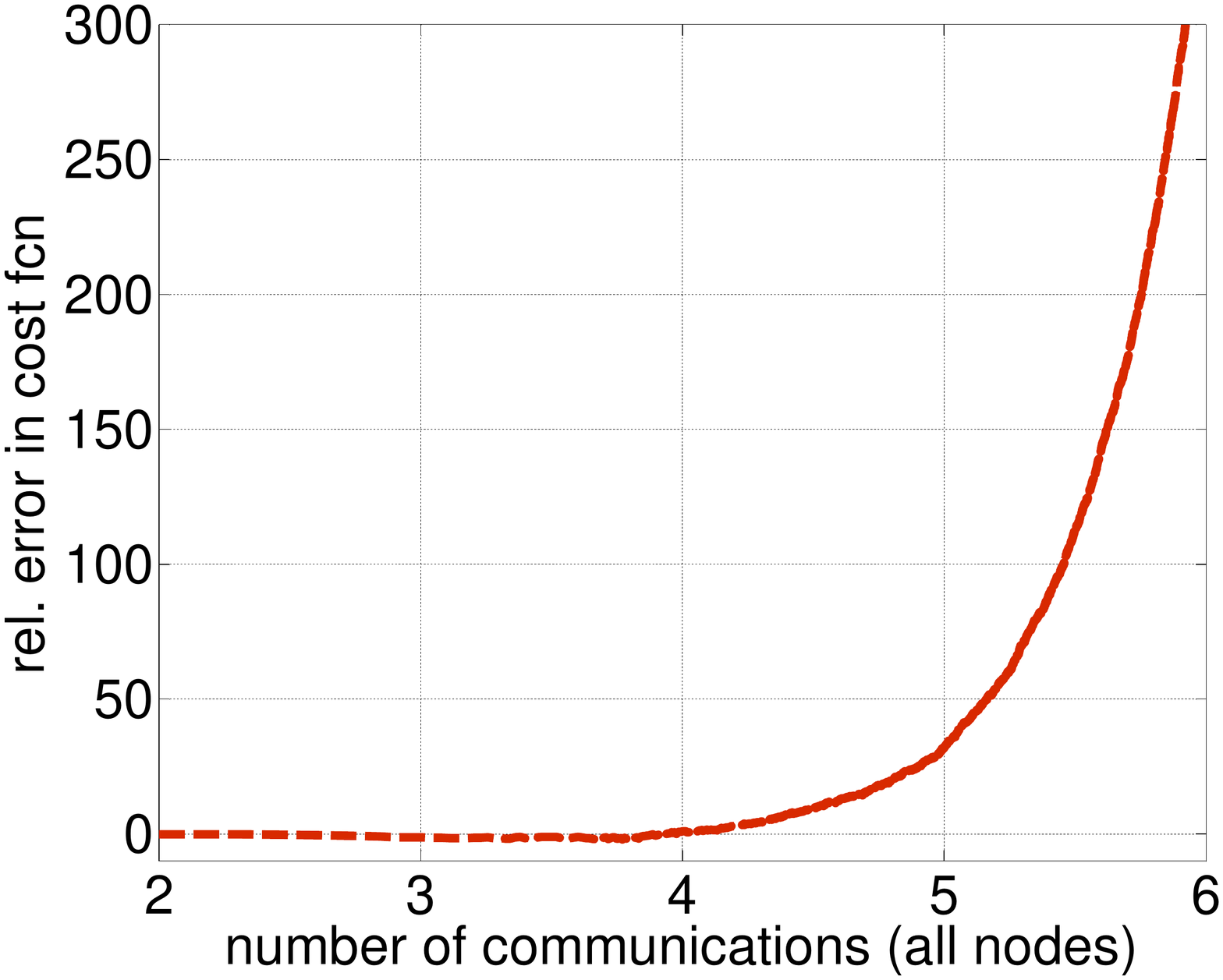}
       \includegraphics[height=1.25 in,width=3.18 in]{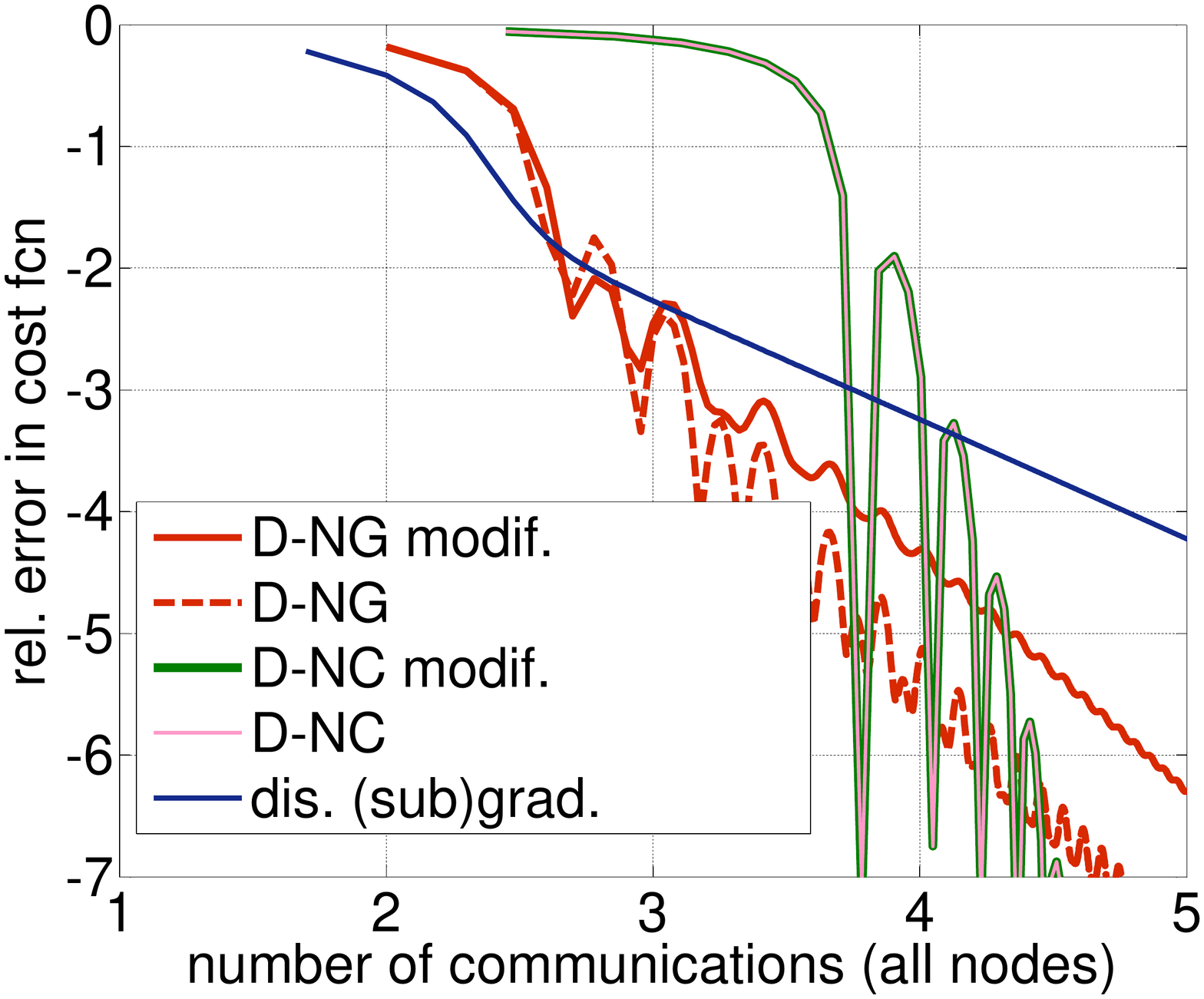}
       \includegraphics[height=1.25 in,width=3.18 in]{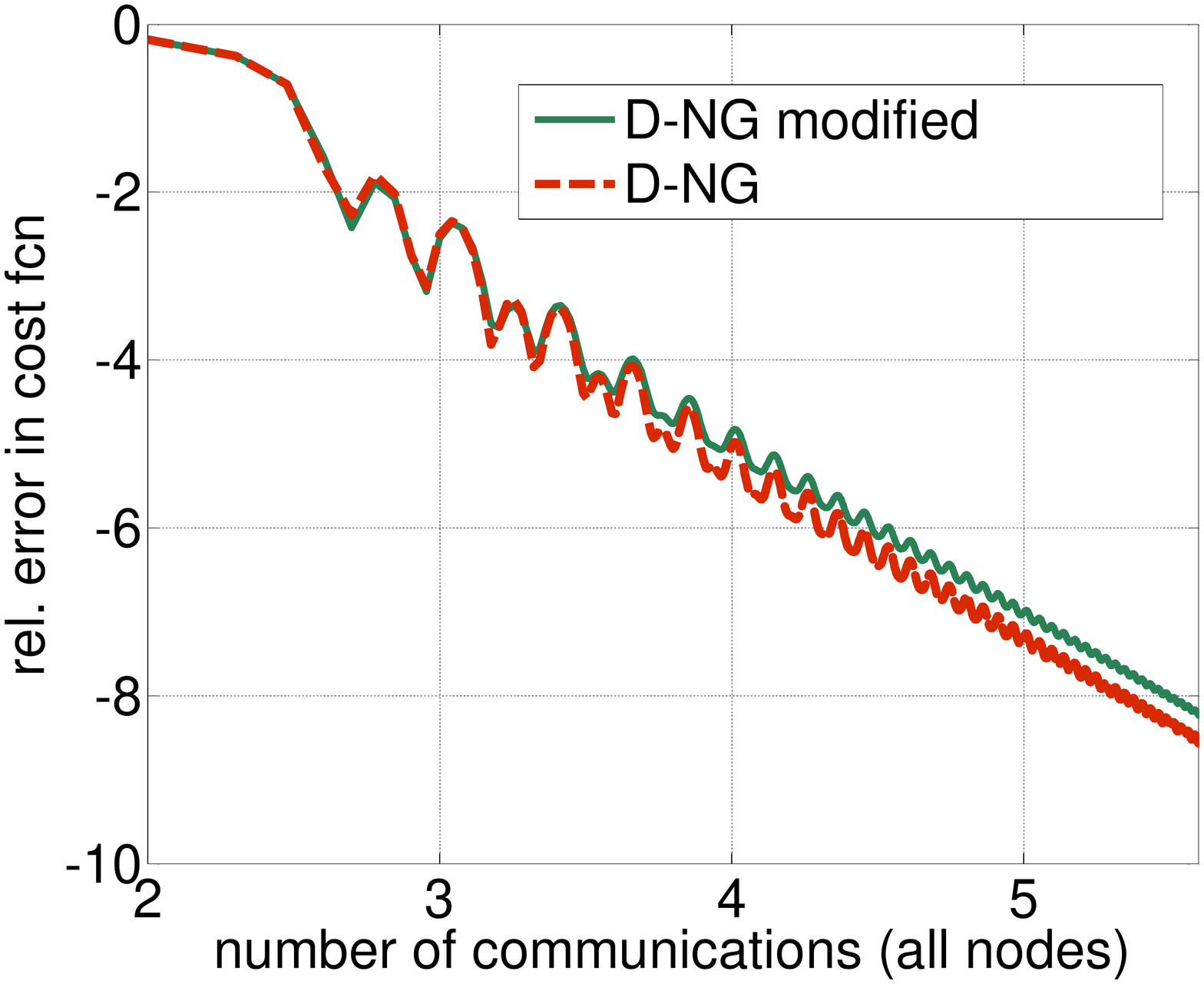}
       \caption{Average normalized optimality gap
  $\mathbf{err_f}$ %=\frac{1}{N}\sum_{i=1}^N \frac{f(x_i)-f^\star}{f(0)-f^\star}$
   vs.~total number $\mathcal K^\prime$
   of scalar transmissions, across all nodes ($\log_{10}\!-\!\log_{10}$ scale.)
    Two top plots: link failures;
    Two bottom plots: Static network, with bottom comparing D--NG and mD--NG
     when $W$ is not positive definite.
    }
      \label{figure-1-link-failures}
\end{figure}
\begin{figure}[thpb]
      \centering
       \includegraphics[height=1.25 in,width=3.18 in]{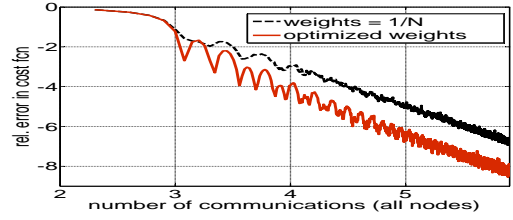}
       \caption{Average normalized optimality gap $\mathbf{err_f}$ vs.~$\mathcal K^\prime$ ($\log_{10}-\log_{10}$ scale)
    for $N=20$-node network and mD--NG method with different weight
    assignments; Red, solid line: optimized weights according to~\cite{weightOpt};
    black, dotted line: $\widehat{W}_{ij}=1/N$, $\forall \{i,j\} \in E.$\label{fig:weightOpt}
    }
      \label{figure-2-link-failures}
\end{figure}

\textbf{Weight optimization}.
Figure~\ref{fig:weightOpt} shows $\mathbf{err_f}$ versus $\mathcal{K}^\prime$ for uniform weights $1/N$ and the optimized weights in~\cite{weightOpt} on a $20$-node, $91$-edge geometric graph (radius $\delta_0=.55$). Links fail independently in time and space with probabilities $P_{ij}=0.5 \times \frac{\delta_{ij}^2}{\delta_0^2}$. The losses are Huber $\theta_i=\pm4(1+\nu_i)$: $+$ for nodes $i=1\cdots7$ and $-$ for $i=8\cdots20$. The $\nu_i$'s are as before. The two plots have the same rates (slopes). The optimized weights lead to better convergence constant  (agreeing with Theorem~\ref{theorem-convergence-rate-D-NG-modified}), reducing the communication cost for the same accuracy.

\section{Conclusion}
\label{section-conclusion-random}
We considered distributed optimization
over random networks where $N$
nodes minimize the sum $\sum_{i=1}^N f_i(x)$
 of their individual convex costs.
 We model the random network
 by a sequence $\{W(k)\}$
  of independent, identically distributed
  random matrices that take values in
  the set of symmetric, stochastic matrices with positive diagonals.
  The $f_i$'s are convex and have Lipschitz continuous
  and bounded gradients.
  We present mD--NG and mD--NC that are resilient to link failures.
   We establish their convergence in terms of the expected optimality gap of the cost function at arbitrary node~$i$:
 mD--NG achieves rates~$O\left(\log k/k\right)$
and~$O\left( \log \mathcal K/\mathcal K \right)$,
where~$k$ is the number of per-node
gradient evaluations and~$\mathcal K$
 is the number of per-node communications; and
 mD--NC has rates $O(1/k^2)$ and $O(1/\mathcal K^{2-\xi})$,
 with $\xi>0$ arbitrarily small.
 % Finally, we extend mD--NC
%  to constrained optimization with compact constraints
%  and establish the same rates.
   Simulation examples
 with link failures and Huber loss functions illustrate our findings.
\section*{Appendix}
\subsection{Proofs of Lemmas~\ref{lemma-decay-moments} and~\ref{lemma-decay-moments-D-NC-modified}}
\label{app:A}
\begin{proof}[Proof of Lemma~\ref{lemma-decay-moments}]
We prove~\eqref{eqn-theorem-phi-2}. For $t=k-1$,
 $\widetilde{\Phi}(k,t)=I$ and~\eqref{eqn-theorem-phi-2} holds.
 Fix $t$, $0 \leq t \leq k-2$.
 For $N \times N$ matrix $A$:
  $\|\widetilde{\Phi}(k,t)\|^2 \leq N \sum_{i=1}^N \left\|  \widetilde{\Phi}(k,t)\,e_i \right\|^2$,
  $e_i$ is the $i$-th canonical vector. Using this, taking expectation:
  \begin{align}
  \label{eqn-A-squared}
  \mathbb E \left[\left\|\widetilde{\Phi}(k,t)^\top \widetilde{\Phi}(k,t)\right\|\right] &= \mathbb E \left[  \left\|  \widetilde{\Phi}(k,t) \right\|^2 \right]\\
  &
  \nonumber
  \leq N \sum_{i=1}^N \mathbb E \left[  \left\| \widetilde{\Phi}(k,t)\,e_i  \right\|^2 \right].
  \end{align}
 Let $\chi_i(s+1):=\widetilde{\Phi}(s,t)e_i=
\widetilde{W}(s+t+2)\cdots \widetilde{W}(t+2) e_i$, $s=0,\cdots ,k-t-2$, $\chi_i(0)=e_i$. Get the recursion:
 \begin{align*}
 \chi_i(s+1)&=\widetilde{W}(s+t+2)\chi_i(s),s=0,\cdots ,k-t-2.
 \end{align*}
 Independence of the $\widetilde{W}(k)$'s
  and nesting expectations:
  \begin{align*}
  &\mathbb E\!\! \left[ \!\|\chi_i(k\!-\!t\!-\!1)\|^2 \! \right]\!\! =\!\!
  \mathbb E \!\left[\!\mathbb E\! \left[\! \left. \chi_i(k\!-\!t\!-\!2)^{\!\!\top}\!\! \widetilde{W}(k)^2\! \chi_i(k\!-\!t\!-\!2)\!\right|\right.\right.\\
  &
  \left.\left.\left| \widetilde{W}(k-1),\cdots ,\widetilde{W}(t+2)\right.\right]\,   \right] \\
  &=
  \mathbb E \left[ \chi_i(k-t-2)^\top \,\mathbb E \left[ \widetilde{W}(k)^2\right]   \chi_i(k-t-2)\right] \\
  &\leq \mathbb E \left[ \left\|\mathbb E \left[ \widetilde{W}(k)^2\right] \right\|\,\left\|  \chi_i(k-t-2)\right\|^2\right] \\
  &\leq
  {\overline{\mu}}^2\,\mathbb E \left[ \left\| \chi_i(k-t-2)\right\|^2\right].
  \end{align*}
  Repeating for $E \left[ \left\| \chi_i(k-t-2)\right\|^2\right]$,
  obtain for all~$i$:
  \begin{align*}
  \mathbb E \left[ \|\widetilde{\Phi}(k,t)\,e_i\|^2  \right] &=
  \mathbb E \left[ \|\chi_i(k-t-1)\|^2  \right]\\
  & \leq ({\overline{\mu}}^2)^{k-t-1}\,\|e_i\|^2 = ({\overline{\mu}}^2)^{k-t-1}.
  \end{align*}
  Plugging this in~\eqref{eqn-A-squared}, \eqref{eqn-theorem-phi-2} follows.
  Next,~\eqref{eqn-theorem-phi} follows from~\eqref{eqn-theorem-phi-2}
  and Jensen's inequality.
  To prove~\eqref{eqn-theorem-phi-3}, consider $0\! \leq\! s\!<\!t\!\leq\! k-2$ ($t\!<\!s$ by symmetry).
    By the independence of the $\widetilde{W}(k)$'s,
    the sub-multiplicative property of norms,
    and taking expectation, obtain:
    \begin{align}
    &\mathbb E\! \left[ \left\| \widetilde{\Phi}(k,s)^\top \widetilde{\Phi}(k,t) \right\|  \right]
    \!\!\leq\!\!
    \mathbb E\! \left[  \|\widetilde{W}(t+1)\cdots \widetilde{W}(s+2)\| \right] \nonumber
    \\
    &
    \nonumber
    \hspace{5.1cm}
    \times \mathbb E \left[ \left\| \widetilde{\Phi}(k,t) \right\|^2  \right] \\
    \label{eqn-apply-Previous}
    &\leq\!\!
    \left(\!N{\overline{\mu}}^{t\!-\!s} \! \right)\!\! \left(\!N^2\!({\overline{\mu}}^2)^{k-t-2}\!  \right)\! =\! N^3\!{\overline{\mu}}^{(k-t-1)+(k-s-1)}.
    \end{align}
    We applied~\eqref{eqn-theorem-phi} and~\eqref{eqn-theorem-phi-2} to get~\eqref{eqn-apply-Previous};
    thus, \eqref{eqn-theorem-phi-3} for $s,t \in \{0,\cdots,k-2\}$.
    If $s=k-1$, $t<k-1$, $\widetilde{\Phi}(k,s)^\top \widetilde{\Phi}(k,t)=\widetilde{\Phi}(k,t)$
     and the result reduces to~\eqref{eqn-theorem-phi-2}. The case $s<k-1$, $t=k-1$
      is symmetric. Finally, if $s=k-1,t=k-1$, the result is trivial. The proof is complete.
\end{proof}
\begin{proof}[Proof of Lemma~\ref{lemma-decay-moments-D-NC-modified}]
We prove~\eqref{eqn-lemma-psi-norm-1}. By~\eqref{eqn-psi-k-t}, $\mathcal W(k)$ is the product of $\tau_k$ i.i.d.~matrices $W(t)$ that obey Assumptions~\ref{assumption-random-weight-matrices} and~\ref{assumption-connected-random}. Hence, by \eqref{eqn-theorem-phi-2}, obtain~\eqref{eqn-lemma-psi-norm-1}:
\begin{align*}
\mathbb E \left[ \left\|\mathcal W(k)\right\|^2 \right] \leq (\overline{\mu}^2)^{\tau_k}
&=N^2\,e^{2\tau_k \log (\overline{\mu})}\\
&
 \leq N^2e^{-2(3\log k+\log N)} = \frac{1}{k^6},
\end{align*}
We \!\!prove\!\! \eqref{eqn-lemma-psi-norm-3}. \!\!Let \!\!$\widetilde{\Psi}(k,t)\!\!\!\!\!:=\!\!\!\widetilde{\mathcal W}\!(k),\!\cdots\!,\widetilde{\mathcal W}(t\!+\!1), \!k \!\!\!\geq \!\!\! t\!\!\!+\!\!\!1$.
For square matrices $A,B$:
$\|\!B^\top\!\! A^\top\!\! A B\|\!\!\! \leq\!\!\! \|A^\top \!\!A\|\|B^\top\!\! B\| = \|B\|^2\,\|A\|^2$.
Applying it $k-t$ times, obtain:
\[
\left\| \widetilde{\Psi}(k,t)^\top \widetilde{\Psi}(k,t)\right\|
\leq \left\| \widetilde{\mathcal W}(k)\right\|^2\,\cdots \,\left\|\widetilde{\mathcal W}(t+1)\right\|^2.
\]
Using independence, taking expectation,
and applying~\eqref{eqn-lemma-psi-norm-1}, obtain~\eqref{eqn-lemma-psi-norm-3}.
 By Jensen's inequality, \eqref{eqn-lemma-psi-norm-2} follows from~\eqref{eqn-lemma-psi-norm-3};
  relation~\eqref{eqn-lemma-psi-norm-4} is proved similarly.
\end{proof}

\subsection{Proof of~\eqref{eqn-second-moment-mD-NG}--\eqref{eqn-second-moment-mD-NC}}
\label{subsection-appendix-moments}
Recall the random graph $G(k)$. For a certain connected graph $G_0$,  $G(k)=G_0$ with probability $p_G>0$.
  This, together with
 $\underline{w}>$ by Assumption~\ref{assumption-random-weight-matrices}, implies that there exists $\mu_4 \in [0,1)$, such that
 $
 \mathbb E \left[ \left\|  \widetilde{W}(k)\right\|^4\right] \leq (\mu_4)^4.
 $
In particular, $\mu_4$ can be taken as:
\[
\mu_4 = (1-p_G) + p_G \left( 1-\underline{w}^2\,\lambda_{\mathrm{F}}\left(G_0\right)  \right)^2<1,
\]
where $\lambda_{\mathrm{F}}\left(G_0\right)$ is the second largest eigenvalue  of
the unweighted Laplacian of $G_0$. Now, consider~\eqref{eqn-star-star-proof-random-trian}.
 Let $\widetilde{f}_k:=\frac{1}{N}(f(\overline{x}(k))-f^\star)$.
 Squaring~\eqref{eqn-star-star-proof-random-trian}, taking expectation, and by the Cauchi-Schwarz
 inequality, get:
 \begin{align}
\nonumber
 &\mathbb E \!\left[\! (\widetilde{f}_k)^2\!\right]\!\!
 \leq \!\!
 \frac{4 R^2}{c^2\,k^2}
 \!+\!
 \frac{4 R L}{N \,c}
 \!\frac{1}{k^2}\!\!\sum_{t=1}^k\!\! \frac{(t+1)^2}{t}\mathbb E\! \left[\! \|\widetilde{y}(t-1)\|^2\! \right] \\
 &
 \nonumber
 +
 \frac{L^2}{N^2\,k^2}\,\sum_{t=1}^k \sum_{s=1}^k \frac{(t+1)^2}{t}
 \frac{(s+1)^2}{s}
\sqrt{ \mathbb E \left[ \, \|\widetilde{y}(t-1)\|^4\,  \right]}\\
&
 \label{eqn-cauchi-square}
\hspace{3cm}
\times \sqrt{\mathbb E \left[ \, \|\widetilde{y}(s-1)\|^4\,  \right]},
 \end{align}
where recall $R:=\|\overline{x}(0)-x^\star\|.$ The first term in~\eqref{eqn-cauchi-square}
is $O(1/k^2)$; by Theorem~\ref{theorem-consensus-D-NG-modified}, the second is $O(\log k/k^2)$.
 We upper bound the third term. Recall~\eqref{eqn-recall-for-appendix},
 let by $\mathcal{U}_k:=\|\widetilde{W}(k)\|$. Fix $s<t$, $s,t \in \{0,1,\cdots ,k-1\}$. By the sub-multiplicative
 property of norms:
 \begin{align*}
 \left\|  \widetilde{\Phi}(k,t)^\top \widetilde{\Phi}(k,s) \right\|
 \!\!\leq \!\!\left(\mathcal{U}_k^2  \mathcal{U}_{k-1}^2 \cdots \mathcal{U}_{t+2}^2\right)
 \!\!
 \left(\mathcal{U}_{t+1}\cdots \mathcal{U}_{s+2}\right)\!.
 \end{align*}
 For $t=s$,
  $
 \left\|  \widetilde{\Phi}(k,t)^\top \widetilde{\Phi}(k,t) \right\|
 \leq \mathcal{U}_k^2  \mathcal{U}_{k-1}^2\cdots \mathcal{U}_{t+2}^2
 $.
Let
 $
 \widehat{\mathcal U}(t,s):=\mathcal{U}_t  \mathcal{U}_{t-1}  \cdots   \mathcal{U}_{s+1}
 $,
 for $t>s$, and $\widehat{\mathcal U}(t,t)=I$.
 Further, let $\widehat{b}(k,t):=\frac{8(k-t-1)(t+1)}{k}+5$. Squaring \eqref{eqn-recall-for-appendix}:
\begin{align}
\nonumber
&\|\widetilde{z}(k)\|^4
\leq   (9 c^4 N^2 G^4)
%
%\sum_{t_1=0}^{k-1}\sum_{t_2=0}^{k-1}
%\sum_{t_3=0}^{k-1}\sum_{t_4=0}^{k-1}
\mathop{\sum\sum\sum\sum}_{\scriptsize t_1,t_2,t_3,t_4=0}^{k-1}
\widehat{b}(k,t_1)\widehat{b}(k,t_2)\\
&
\label{eqn-z-fourth-moment}
\hspace{4cm}
\times
\widehat{b}(k,t_3)\widehat{b}(k,t_4) \\
&\left(\widehat{\mathcal{U}}(k,t_1+1)\right)^{\!\!4}\!\!\! \left(\widehat{\mathcal{U}}(t_1+1,t_2+1)\right)^{\!\!3}\!\!\!
\left(\widehat{\mathcal{U}}(t_2+1,t_3+1)\right)^{\!\!2} \nonumber \\
&
\left(\widehat{\mathcal{U}}(t_3+1,t_4+1)\right)^4\frac{1}{(t_1+1)(t_2+1)(t_3+1)(t_4+1)}. \nonumber
\end{align}
Fix $t_i\!\! \in\!\! \{0,\cdots ,k-1\}, i=1,\cdots,4$,
$t_1\geq t_2 \geq t_3 \geq t_4$.
By independence of $\mathcal{U}_k$'s,
 $\left(\!\mathbb E\!\!\left[\!\mathcal{U}_k^j\!\right]\!\right)^{4/j}
\!\!\! \!\leq\!\mathbb E\!\!\left[\! \mathcal{U}_k^4\! \right]$, $j=1,\!2,\!3$:
\begin{align}
\label{eqn-app-apply}
&\mathbb E \left[ (\widehat{\mathcal{U}}(k,t_1+1))^4(\widehat{\mathcal{U}}(t_1+1,t_2+1))^3\right.\\
&
\nonumber
\hspace{2cm}
\left.
(\widehat{\mathcal{U}}(t_2+1,t_3+1))^2 (\widehat{\mathcal{U}}(t_3+1,t_4+2))\right]
\\
&\leq\!\!\!
(\mu_4^4)^{k\!-\!t_1\!-\!1}\!\! (\mu_4^3)^{t_1\!-\!t_2}\!\! (\mu_4^2)^{t_2\!-\!t_3}\!\! (\mu_4)^{t_3\!-\!t_4}\!\!
=\!\!
(\mu_4)^{\sum_{i=1}^4 (k-t_i)-1}. \nonumber
\end{align}
Taking expectation in~\eqref{eqn-z-fourth-moment}, and applying~\eqref{eqn-app-apply}, we get:
\begin{align}
\nonumber
\mathbb E \left[  \|\widetilde{z}(k)\|^4 \right]
&\leq
(9 c^4 N^2 G^4)\left( \sum_{t=0}^{k-1}\widehat{b}(k,t)\mu_4^{k-t}(t+1)^{-1}  \right)^4 \\
&
\label{eqn-apply-this-now}
= O(1/k^4),
\end{align}
where the last equality uses Lemma~\ref{lemma-bounds-on-sums}. Applying~\eqref{eqn-apply-this-now} to~\eqref{eqn-cauchi-square},
the third term in~\eqref{eqn-cauchi-square} is $O(\log^2 k / k^2)$.
Thus, $\mathbb E \left[ (\widetilde{f}_k)^2\right]=O(\log^2 k/k^2)$.
  Express
  $f(x_i(k))-f^\star = N \widetilde{f}_k + (f(x_i(k))-f(\overline{x}(k)))$,
  and use $(f(x_i(k))-f^\star)^2 \leq 2 (N \widetilde f_k)^2 + 2(f(x_i(k))-f^\star)^2
   \leq 2 (N \widetilde f_k)^2 + 2 G N \|\widetilde{x}(k)\|^2 $, where the
   last inequality follows by Assumption~\ref{assumption-bdd-gradients}.
   Taking expectation,
   applying Theorem~\ref{theorem-consensus-D-NG-modified},
   and using \eqref{eqn-apply-this-now}, the result \eqref{eqn-second-moment-mD-NG} follows.
   For mD--NC, prove~\eqref{eqn-second-moment-mD-NC} like~\eqref{eqn-second-moment-mD-NG} by letting $\tau_k=\lceil  \frac{3 \log k}{- \log \mu_4}\rceil$.

\bibliographystyle{IEEEtran}
\bibliography{IEEEabrv,bibliographyNovoApr27}
\end{document}